\documentclass[reqno]{amsart}
\usepackage{graphicx}

\usepackage{enumitem}
\usepackage[numbers,sort&compress]{natbib}
\usepackage[foot]{amsaddr}
\usepackage[final,color,notref,notcite]{showkeys}
\usepackage{booktabs}
\usepackage{dsfont}
\usepackage{xcolor}
\usepackage{colortbl}
\usepackage{multirow,array}
\usepackage{slashed}

\usepackage{hyperref}
\hypersetup{colorlinks,citecolor=green}
\usepackage[capitalise]{cleveref}
\definecolor{labelkey}{rgb}{0,.56,.7}
\definecolor{tablerows}{rgb}{0.906, 0.8, .949}

\usepackage[charter]{mathdesign}
\usepackage[scaled=.96,lf]{XCharter}

\usepackage{algorithm}
\usepackage{algpseudocode}

\theoremstyle{plain}
\newtheorem{thm}{Theorem}

\newtheorem{cor}[thm]{Corollary}
\newtheorem{lem}[thm]{Lemma}
\newtheorem{prop}[thm]{Proposition}

\theoremstyle{definition}
\newtheorem{defi}{Definition}
\newtheorem*{exa}{Example}

\theoremstyle{remark}
\newtheorem*{rem}{\bf Remark}

\let\x\times

\let\ol\overline
\mathchardef\myhyph="2D

\makeatletter
\newcommand{\looongrightarrow}[1][]{\ext@arrow 0099\rightarrowfill@{#1}}
\makeatother

\def\bbZ{\mathbb{Z}}
\def\bbR{\mathbb{R}}
\def\bbN{\mathbb{N}}

\def\bbE{\mathbb{E}}
\def\bbP{\mathbb{P}}

\def\L{\mathcal{L}}

\DeclareSymbolFont{Eulerscripteusm10}{U}{eus}{m}{n}
\SetSymbolFont{Eulerscripteusm10}{bold}{U}{eus}{b}{n}
\DeclareMathSymbol{\euB}{\mathord}{Eulerscripteusm10}{"42}
\DeclareMathSymbol{\euC}{\mathord}{Eulerscripteusm10}{"43}
\DeclareMathSymbol{\euF}{\mathord}{Eulerscripteusm10}{"46}
\DeclareMathSymbol{\euI}{\mathord}{Eulerscripteusm10}{"4A}
\DeclareMathSymbol{\euK}{\mathord}{Eulerscripteusm10}{"4B}
\DeclareMathSymbol{\euL}{\mathord}{Eulerscripteusm10}{"4C}
\DeclareMathSymbol{\euQ}{\mathord}{Eulerscripteusm10}{"51}
\DeclareMathSymbol{\euR}{\mathord}{Eulerscripteusm10}{"52}
\DeclareMathSymbol{\euS}{\mathord}{Eulerscripteusm10}{"53}
\DeclareMathSymbol{\euX}{\mathord}{Eulerscripteusm10}{"58}
\DeclareMathSymbol{\euZ}{\mathord}{Eulerscripteusm10}{"5A}

\DeclareMathSymbol{\rH}{\mathord}{Eulerscripteusm10}{"48}

\def\df{\overset{\mathrm{df}}{=}}

\newcommand{\wt}{\mathop{{\mathrm{wt}}}\nolimits}

\newcommand{\rank}{\mathop{{\mathrm{rank}}}\nolimits}
\newcommand{\im}{\mathop{{\mathrm{im}}}\nolimits}

\newcommand{\inter}{\mathop{{\mathrm{int}}}\nolimits}

\newcommand{\id}{\mathop{{\mathrm{id}}}\nolimits}

\def\a{\alpha}
\def\b{\beta}
\def\g{\gamma}
\def\d{\delta}

\def\la{\lambda}

\def\bbeta{\boldsymbol{\b}}

\allowdisplaybreaks

\makeatletter
\renewcommand\subsection{\@startsection{subsection}{2}%
  \z@{.5\linespacing\@plus.7\linespacing}{.5\linespacing}%
  {\normalfont\scshape}}
\renewcommand\subsubsection{\@startsection{subsubsection}{3}%
  \z@{.5\linespacing\@plus.7\linespacing}{.5\linespacing}%
  {\normalfont\scshape}}
\makeatother

\begin{document}

\title{CSS codes from the Bruhat order of Coxeter groups}

\author{Kamil Br\'adler}

%\address{}

%\date{\today}

\begin{abstract}
I introduce a  method to generate families of CSS codes with interesting code parameters. The object of study is Coxeter groups, both finite and infinite (reducible or not), and a geometrically motivated partial order of Coxeter group elements named after Bruhat. The Bruhat order is  known to provide a link to algebraic topology -- it doubles as a face poset capturing the inclusion relations of the $p$-dimensional cells of a regular CW~complex and that is what makes it interesting for QEC code design. Assisted by the  Bruhat face poset interval structure unique to Coxeter groups I show that the corresponding chain complexes can be turned into multitudes of CSS codes. Depending on the approach, I obtain  CSS codes (and their families) with controlled stabilizer weights, for example $[6006, 924, \{{\leq14},{\leq7}\}]$ (stabilizer weights~14 and 9) and $[22880,3432,\{{\leq8},{\leq16}\}]$ (weights 16 and 10), and CSS codes with highly irregular stabilizer weight distributions such as $[571,199,\{5,5\}]$. For the latter, I develop a weight-reduction method to deal with rare heavy stabilizers. Finally, I show how to extract four-term (length three) chain complexes that can be interpreted as CSS codes with a metacheck.
\end{abstract}

\maketitle

%\begin{asydef}
%import graph;
%import settings;
%\end{asydef}

\thispagestyle{empty}

\section{Introduction}\label{sec:intro}

Search and study of new quantum error-correction (QEC) codes has been a significant part of quantum computing theory research since its early days. Stabilizer codes~\cite{gottesman1997stabilizer} as a quantum equivalent of classical linear codes are the main object of interest together with an important subclass of CSS codes~\cite{steane1996multiple,calderbank1996good}, which is a type of QEC codes I focus here too. The discovery of new codes is of a great theoretical interest but nowadays the main motivator is the quest for a utility-scale universal fault-tolerant quantum computer pursued (mostly) by private sector. Physical and technological constraints of the prospective quantum computing platforms dictate the desired properties of the used codes and their families. There is  a very long list of the code properties that needs to be evaluated  from the scalability point of view with regard to a given platform before a decision is taken to implement it in the quantum hardware and software. If the main figure of merit is just  an asymptotically constant encoding rate and linearly scaling  distance while the stabilizer weight is bounded then the answer points to good quantum LDPC (low-density parity check) codes whose (existential) discovery is one of the most important advances in recent quantum computing~\cite{panteleev2022asymptotically,leverrier2022quantum,dinur2023good,evra2022decodable,fawzi2020constant,breuckmann2021balanced}. Whether they turn out to be practically usable code families in a finite regime is so far an open question. While waiting for the answer, researchers came up with constructive families~\cite{breuckmann2021quantum,vasic2025quantum} of `pretty good'  finite-size LDPC codes (such as codes with a constant or slowly decreasing rate and suboptimal but guaranteed distance scaling) with solid distances, decently bounded stabilizer weights~\cite{kovalev2012improved,tillich2013quantum,zeng2019higher,panteleev2021degenerate,breuckmann2021balanced,hastings2021fiber,wang2022distance,lin2024quantum,leverrier2025small} and with a growing focus on the scalability of performing logical operations in a fault-tolerant manner. In this regard they already started to outperform the surface code~\cite{kitaev2003fault,bravyi1998quantum,dennis2002topological} -- the gold standard of QEC codes one might call a `not-so-bad' LDPC code.

Even for current best pretty good LDPC codes it remains to be seen how practically useful they are on a very large scale. As already mentioned, the distance and rate scaling is not the only performance indicator to watch from the scalability point of view. There is a considerable interest to come up with new code  creation methods even at the cost of the code families not necessarily  being LDPC~\cite{jochym2014using,yoshida2025concatenate,hastings2016quantum,yamasaki2024time,goto2024high}. My proposal, which is based on the properties of Coxeter groups, falls into this category. Coxeter groups are a family of discrete groups. They figure prominently in various areas of mathematics but their origin is geometrical as a formalization of a (finite) group of reflections in Euclidean geometry~\cite{davis2012geometry,coxeter1973regular,humphreys1990reflection}. All finite irreducible Coxeter groups were classified and they were further generalized to affine Euclidean spaces as well as  hyperbolic spaces of various signatures. The reflection hyperplanes (as abstracted mirrors) can be shifted away from the origin in affine vector spaces, giving rise to tessellations. Consequently, Coxeter groups for affine reflection systems  (and hyperbolic as well) are infinite.

Coxeter groups provide various mechanisms to generate CSS codes. Compactified affine Euclidean spaces $\bbE^d$ give rise to the $d$-dimensional toric codes~\cite{dennis2002topological,jochym2021four} and their twisted variants. Their asymptotic properties pale in comparison with the variety and properties of the CSS codes one can get from tesselated hyperbolic spaces~\cite{guth2014quantum,breuckmann2016constructions}. A different construction is 2D and 3D color codes~\cite{bombin2006topological,landahl2011fault}, which are again CSS codes that can be derived from regular (or less regular) tessellations. Color codes were originally introduced without referring to Coxeter groups but this philosophy, adopted by many authors~\cite{vasmer2019fault,vasmer2019three,jochym2021four}, is quite fruitful and leads to multiple generalizations~\cite{huang2023tailoring}. Finally, a recent work specifically targeting Coxeter groups introduced a different CSS code construction based on assigning stabilizers to the facets of polytopes constructed from Cayley graphs of Coxeter groups~\cite{coble2026coxeter}. This approach has a more combinatorial flavor and is probably closest to the point of view adopted here of using combinatorial rather than geometrical properties of Coxeter groups.

The construction introduced in this paper is yet another way of extracting CSS codes from Coxeter groups. The structure of Coxeter groups admits several different partial orders and one of them is called the Bruhat order. It stands out thanks to its rich properties encompassing combinatorics, geometry and algebraic topology. In particular, every Bruhat poset (partially ordered set) can be interpreted as a cell complex known as a regular CW complex. More precisely, the partial order relation of the Bruhat poset is in fact a face poset of closed cells of a certain manifold ordered by inclusion. The manifold is a cellulated $d$-dimensional sphere $S^d$ that comes with a remarkable internal structure. The key discovery~\cite{bjorner2006combinatorics} is the internal poset structure for the intervals of length two, three and four of any Coxeter group\footnote{Length four intervals were classified just for finite Weyl groups~\cite{hultman2003combinatorial} but rather than the  classification itself it is mainly its existence that is used here.} -- just what is necessary to decompose an arbitrary CSS code. I exploit it to generate CSS codes which are encoding qubits with high rates and respectable distances.  The internal structure nicknamed  diamond- ($S^0$), $k$-crown ($S^1$) and $S^2$-sphere decomposition is then an efficient indicator where to apply a transformation I call splicing. This operation, despite being quite a primitive stabilizer transformation, manages to produce interesting CSS codes from wildly different Coxeter groups -- both finite and infinite -- at the cost of uneven and sometimes heavy stabilizer checks\footnote{I believe that the unique `microscopic' insight into the code structure provided by the sphere decomposition could be exploited more efficiently by, for example, more advanced stabilizer transformation methods.}. I take the first steps to tame it by introducing a stabilizer weight-reduction method~\cite{hastings2021quantum}, which as a side result  can be applied to an arbitrary CSS code.  The final result is a procedure that I call chain complex folding which allows to convert longer Bruhat posets into length two and three chain complexes. The former are again CSS codes but this time with a better stabilizer weight control. The latter are CSS codes equipped with a single metacheck. Almost all codes are conjectured to be part of large finite or infinite families of codes by closely following the growth of the underlying Coxeter groups.

All CSS codes showed as examples have  their distance probabilistically upper-bounded by QdistRnd~\cite{QDistRnd,pryadko2023qdistrnd}, which is  implemented as a GAP~\cite{GAP4} library. Smaller or low-distance codes are then exactly calculated by dist-m4i\footnote{\href{https://github.com/QEC-pages/dist-m4ri}{dist-m4ri\label{fn:distm4ri}}} also developed by L.~Pryadko and collaborators and confirmed by another exact method used by Qubitserf\footnote{\href{https://github.com/qubitserfed/Qubitserf}{Qubitserf\label{fn:qubitserf}}} developed by \textcommabelow{S}.~Cercelescu. I used both exact methods except for edge cases where either I lack time or memory. Everything else remains upper-bounded (Qubitserf has a multithread option but it declares segmentation fault for larger codes -- the tool is under development). Despite the differences in the underlying algorithms  the probabilistic ones often agree with the deterministic ones and that gives me more confidence that they are accurate even for large codeblock examples. In those cases I went to great lengths to corroborate the results by repeated sampling as well as by deriving (by yet another method)  the actual logical Pauli operators and finding the one(s) with the smallest weight. This obviously does not provide a provable lower bound but it serves as a sanity check and ultimately increases the confidence in the actual estimates. That being said, some of the discovered codes are very large and none of the distance-estimation methods  allows me to claim that the chances of lowering the reported distance values are beyond reasonable doubt. In fact, I would be able to generate much larger codes based on Bruhat order but the computational cost  to get even a reasonably confident upper bound by~\cite{QDistRnd} is just too high.

This is the first caveat of the paper. The second one is the detailed performance and iterative (asymptotic and finite) behavior of the introduced weight reduction procedure. The $S^{1(2)}$~Bruhat-based codes are quite often high-rate/high-weight and the weight-reduction method is capable of reducing the rate $k/n$ by repeatedly increasing the physical qubit number while the hope (but not the proof) based on numerical simulations is that the distance is non-decreasing. If true it is an economic way of decreasing weight and keeping the codes' good parameters. So there is a reason for optimism but its confirmation is a major project on its own and therefore deferred for future investigations.

The paper consists of the following main parts:~\cref{sec:CSScodes} introduces the principal method and its variants to obtain non-trivial CSS codes with interesting parameters described as CW chain complexes of length three and four. The former is a CSS code while the latter is a CSS code equipped with a metacheck.~\cref{sec:weightLoss} introduces a weight-reduction method to deal with a negative byproduct of the code creation procedure, where some checks become too heavy. The paper leaves plenty of open questions summarized in Conclusions,~\cref{sec:concl}, followed by~\cref{sec:Apps} in the form of several Appendices, where I collected  useful facts about homology and chain complexes with perhaps less frequented but terminology-heavy topics such as posets and regular CW complexes. This is for the sake of completeness even though the language of algebraic topology is now a standard tool in QEC research,  having been used explicitly  at least since the work of  Freedman and Meyer~\cite{freedman2001projective} (implicitly, of course, already in~\cite{kitaev2003fault} posted years earlier), see also~\cite{bombin2007homological}.

\section{CSS code construction}\label{sec:CSScodes}

Every CSS code  discovered in this paper starts its life as an open interval $(w_b,w_t)$ of a Coxeter system $(W,S)$, which is a graded face poset $\euF(S^d)$ capturing the inclusion relation of the $p$-cells of a regular CW decomposition of a $d$-dimensional sphere~$S^d$, see Appendices  in~\cref{sec:Apps} for more detail.  The face poset $\euF(S^d)$ is a CW $(d+2)$-chain  complex and for a sufficiently big~$d$ it gives rise to shorter chain complexes of interest. These structures are at the core of the methods presented here so let's formalize it.
\begin{defi}\label{def:BruhatChains}
  Let $(W,S)$ be any Coxeter system and $\euB=(w_b,w_t)$ be an open interval of the Bruhat order for $w_b,w_t\in W$. Thanks  to~\cref{thm:CWcellulationBruhat}, $\euB=\euF(S^d)$, where $\ell(\euB)=d+2$ is the length of~$\euB$.  I denote  $(W,w_b,w_t,p)_{2k+1}$  to be a subposet of $ \euF(S^d)$ in the form of $2k+1$ layers $\{l_p\}\in\euF(S^d)$  of length~$p$ for $-k\leq p\leq k$ whenever $\ell(w_b)\leq p-k$ and $p+k\leq\ell(w_t)$.
\end{defi}
Every  $(W,w_b,w_t,p)_{2k+1}$ is a face subposet of length $2k+1-1=2k$ consisting of $2k+1$ layers $\{l_i\},p-k\leq i\leq p+k$. For $k=1$ its Hasse diagram is called a quantum Tanner graph because it is a CSS code even though a trivial one -- encoding zero logical qubits (see~\cref{cor:CSStriplesTriv}). I will use the language of Tanner graphs and face posets interchangeably but I will also need to use the chain complex terminology. Here it might get a little bit confusing because $(W,w_b,w_t,p)_{2k+1}$ is formally not a chain complex. Why? Finite chain complexes have two distinguished zero modules as their initial and terminal modules. If I consider an  entire face poset $\euF(S^d)$  the zero modules are naturally present in the form of $w_b$ and $w_t$.  If, on the other hand, a face subposet  $(W,w_b,w_t,p)_{2k+1}$ of length $2k$ is extracted from $\euF(S^d)$ it is formally necessary to add common bottom and top elements that I will denote $\hat{b}$ and $\hat{t}$ in order to promote it to a chain complex of length not $2k$ but rather $2k+2$, that is, consisting of $2k+1+2=2k+3$ layers. This chain complex will be denoted $(W,w_b,w_t,p,\hat{b},\hat{t})_{2k+1}$ when $\hat{b},\hat{t}$ play an important role but otherwise I abuse notation and also call $(W,w_b,w_t,p)_{2k+1}$ a chain complex. The practical difference when  describing how to obtain non-trivial CSS codes later in the text (codes having non-zero number of logical qubits) is null. An added advantage of introducing $\hat{b},\hat{t}$ is that $[\hat{b},\hat{t}]=(W,w_b,w_t,p,\hat{b},\hat{t})_{2k+1}$ is called a closed poset interval (\cref{app:partialOrder}) and
$$
(\hat{b},\hat{t})=(W,w_b,w_t,p)_{2k+1}
$$
is an associated open poset interval. So even though the number of layers between $[\hat{b},\hat{t}]$ and $(\hat{b},\hat{t})$ differs, the interval length of the posets is equal, that is  $\ell([\hat{b},\hat{t}])=\ell((\hat{b},\hat{t}))=2k+2$.

Longer chain complexes  $(W,w_b,w_t,p,\hat{b},\hat{t})_{2k+1}$  (or $(W,w_b,w_t,p)_{2k+1}$)  for $k=2,3$ will be used as well as an intermediary step to again get Tanner graphs of non-trivial CSS codes. The main construction in this paper is base on a special property of the Bruhat poset of Coxeter groups which is, in general, not shared by other discrete groups: They can be decomposed into elementary spheres\footnote{$S^d$~sphere decompositions for $d>2$ are possible but they can't be called elementary and mainly they do not seem to be known through classification.}, namely $S^0$, $S^1$ and $S^2$.

Before discussing the code-making strategies it seems reasonable to first show a few examples of Coxeter groups and their Bruhat order posets to get a handle on the poset terminology and what is going on here. I chose three examples that do not provide interesting codes due to their size but  to me the evidence suggests that they are a part of different (infinite) code families and I will show that their bigger siblings do provide interesting codes.
\begin{exa}
  Coxeter system $(A_3,\{s_1,s_2,s_3\})$: An infinite class of Coxeter groups called $A_n$ is better known as the symmetric groups $S_{n+1}$ of order $(n+1)!$. $A_n$ is one of the most important classes of discrete groups whose various properties and links to diverse areas of mathematics have been investigated in detail. The Coxeter matrix for $n=3$ is in~\cref{eq:CoxMatA3} and the Bruhat poset is depicted in Fig.~\ref{fig:A3BruhatPoset}. This is the only example where I will  illustrate in detail the rich poset terminology introduced in Appendix~\ref{app:partialOrder} followed by several sphere decompositions in~\cref{app:sphereClassification}.
  \begin{figure}[t]
      \resizebox{13cm}{!}{\includegraphics{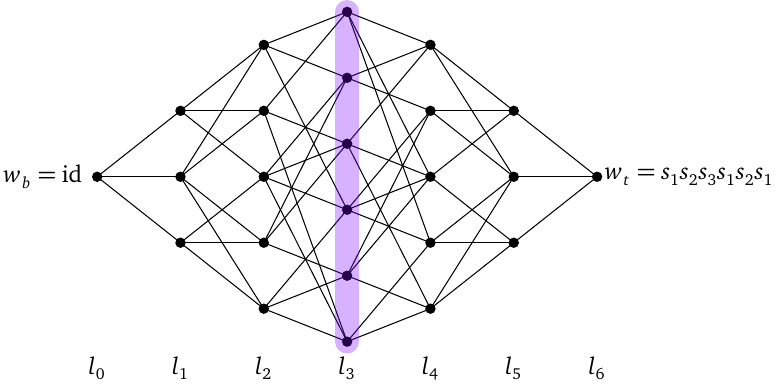}}
      \caption{The Bruhat order of the Weyl group~$A_3$,~\cref{eq:A3presentation}. The open interval $(w_b,w_t)$ is a face poset of a regular CW-cellulated sphere $S^d$ for $d=4$. The vertices can be seen as $(p-1)$-cells ($1\leq p\leq d+1=5$) ordered by inclusion or as Coxeter group elements $w=s_{i_1}\dots s_{i_p}\in A_3$ of the reduced length $\ell(w)=p$ related by a reflection. The edges are covering relations of the Bruhat order.}
      \label{fig:A3BruhatPoset}
  \end{figure}

  The Bruhat order on $A_3$ is shown in the form of an undirected Hasse diagram in~\cref{fig:A3BruhatPoset}. The vertices of the graph are the group elements and the undirected edges are the covering relations. The lowest (bottom) element~$w_b=\id$  is on the far left and top element $w_t=s_1s_2s_3s_1s_2s_1$ is the rightmost vertex (using the toppled convention). Hence the depicted poset is a closed interval $[w_b,w_t]$. The poset is graded by the reduced word length function $\ell$, where $\ell(\id)=0$ and $\ell(w_t)=6$ already in its reduced form and so  $\ell([w_b,w_t])=6$. Since $|A_3|=24$ the entire poset can de drawn but if necessary nothing prevents me from choosing a different pair $w_b,w_t$ to get a subposet as long as $w_b<w_t$ holds. The poset's Hasse diagram is stratified by (vertical) layers $l_i$ of all incomparable poset elements of rank~$i$. So, for example, the number of layer elements for $i=3$ is $|l_3|=6$ highlighted as a purple line.

  In the CW interpretation of this poset, the open interval $(\id,s_1s_2s_3s_1s_2s_1)$ is a regular CW decomposition of the $S^4$~sphere, where each layer~$l_p$,  $1\leq p\leq 5$, is occupied by $|l_p|$ $(p-1)$-cells ($(p-1)$-faces), see~\cref{app:CWcomplexes} and~\ref{app:sphereClassification}. The covering relation becomes a closed cell inclusion and the poset therefore becomes a face poset~$\euF(S^4)$. Every triple of layers $\{l_i\}_{i=p-1,p,p+1}$ is a Tanner graph of a CSS code for all $2\leq p\leq 4$, where the layer $l_p$ is assigned to data qubits and $l_{p\pm1}$ are $X$ and $Z$ stabilizers or vice versa. This avoids the left- and rightmost triple ($p=1$ and $p=5$, respectively) which, despite being legitimate CSS codes, are not practically as interesting as the rest due to the presence of a single vertex (stabilizer). The triples $\{l_i\}$ I just described are examples of $(W,w_b,w_t,p)_{3}$ from~\cref{def:BruhatChains} as the building blocks of non-trivial CSS  codes introduced later.

  The Hasse diagram of the Bruhat poset of a Coxeter group often has a geometrical interpretation. Although nowhere in this work I use any such construction it is worth mentioning it in the case of $A_3$. By comparing the left or right Cayley graphs based on the generating set consisting of the simple reflections $s_i$ (see~\cref{app:discreteGroups}) with the additional general reflections $t_j$ responsible for the Bruhat poset covering relations (see~\cref{eq:reflectionT}) it follows that the Bruhat order uses Cayley graphs as a scaffold for any Coxeter group. So, given  a Coxeter group's Cayley graph with the corresponding weak order (see~\cref{app:discreteGroups}) the Bruhat order on the same group is obtained by adding some more edges to the corresponding Cayley Hasse diagram. The left (right) Cayley graph of $A_n$ tessellates~$\bbR^n$.  For $n=3$ it is called a permutahedron or a truncated octahedron and it made appearance in several QEC code constructions~\cite{watson2015qudit,huang2023tailoring}. I  depict how the `Bruhat boosted' permutahedron looks like in~\cref{fig:A3poset3D}, where one can see additional edges inside the permutahedron and also inside the hexagonal faces -- all coming from the non-simple reflections. Similar enhancements are possible for other Coxeter groups whose Cayley graphs have a geometrical interpretation (such as affine and hyperbolic tessellations). 

  \begin{figure}[t]
  \resizebox{5.6cm}{!}{\includegraphics{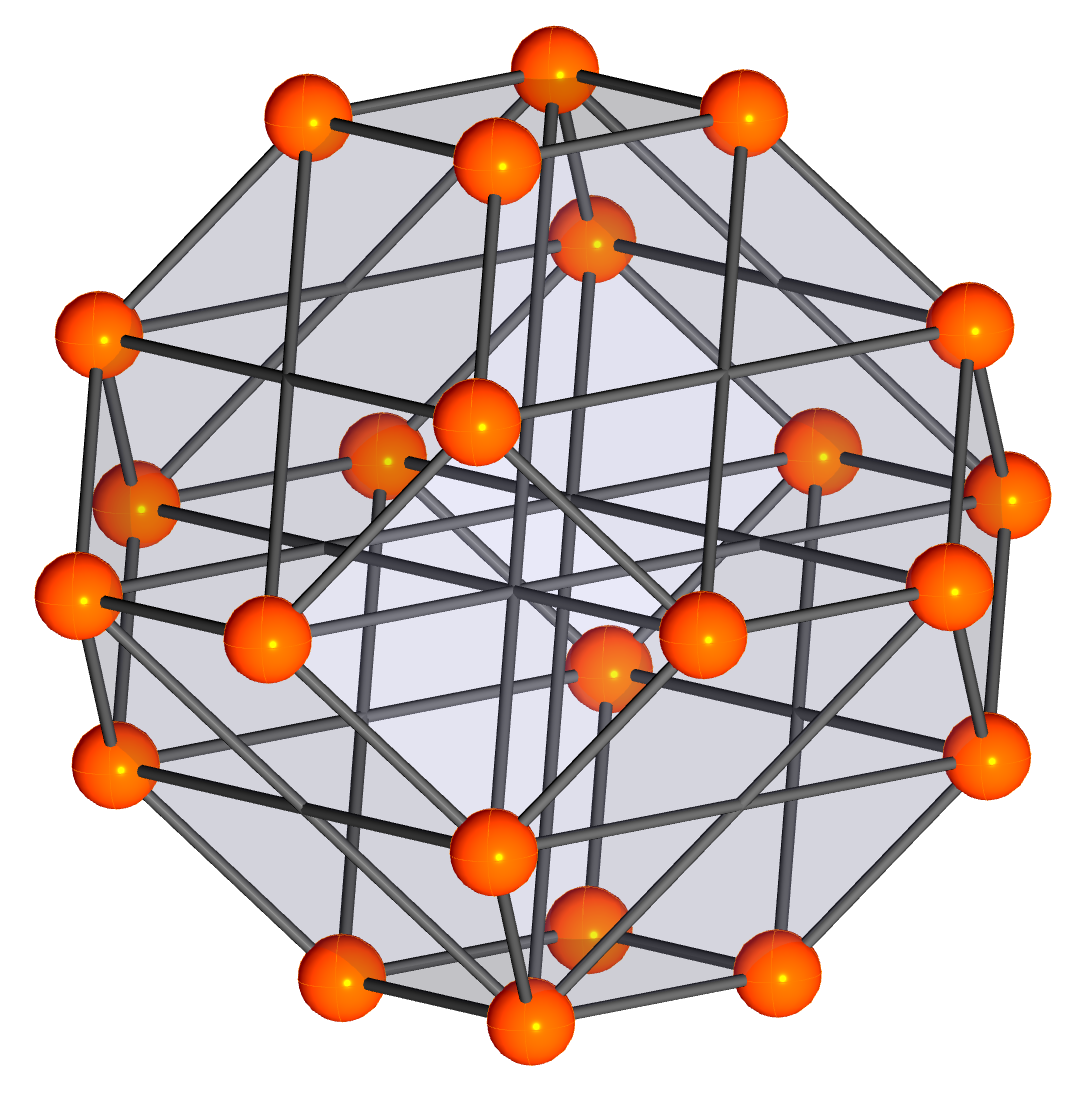}}
  \caption{The Bruhat poset of $A_3$ from~\cref{fig:A3BruhatPoset} in a geometrically friendly way as a truncated octahedron (also known as a permutahedron in combinatorics) with some additional inner-face edges whose origin are non-simple reflections. The orange balls are the  black vertices and the rods connecting them are the covering relations.}
  \label{fig:A3poset3D}
  \end{figure}

  \begin{figure}[b]
  \resizebox{9cm}{!}{\includegraphics{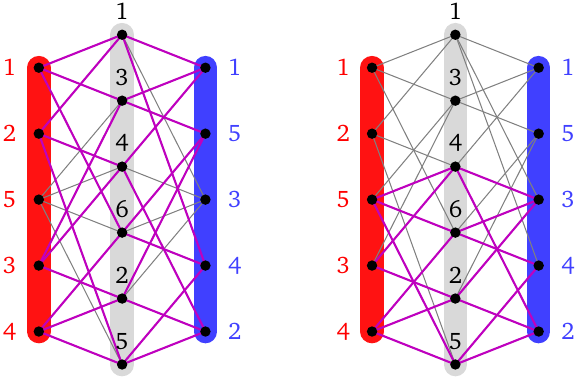}}
  \caption{The Tanner graphs of the $(A_3,\id,s_1s_2s_3s_1s_2s_1,3)_{3}$ CSS code (consisting of layers $l_{2},l_3$ and $l_4$ in~\cref{fig:A3BruhatPoset}) with different highlighted $S^2$ posets in magenta. The connected bottom and top elements $\hat{b}$ and $\hat{t}$ from neighboring layers are not depicted. The corresponding CW complexes are depicted in~\cref{fig:A3S2decompositionCW}.}
  \label{fig:A3S2decompositionTanner}
  \end{figure}

  \begin{figure}[t]
  \resizebox{8.6cm}{!}{\includegraphics{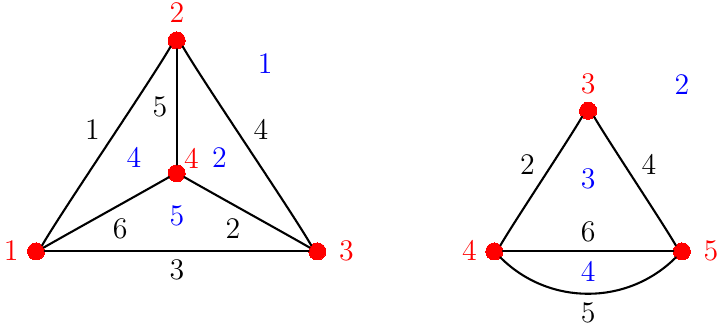}}
  \caption{$S^2$ regular CW complexes corresponding to the magenta subposets in Fig.~\ref{fig:A3S2decompositionTanner} with the same color coding and numbering: red vertices (say $X$ checks), black edges (data qubits) and blues faces ($Z$~checks) including the ambient ones. The CW complex on the left is recognizable as a tetrahedron. Also cf. the classified $S^2$ CW complexes number 1 and 3 in~\cref{fig:hultman}.}
  \label{fig:A3S2decompositionCW}
  \end{figure}

  Let's turn the attention to exemplifying the sphere decomposition -- a feature unique to Coxeter groups as far as the author is aware. The spheres most relevant to any 3-layer poset representing a trivial CSS code are $S^0, S^1$ and $S^2$, see~\cref{fig:S012}. Let's choose a specific triple $\{l_i\}_{i=2,3,4}$ in~\cref{fig:A3BruhatPoset} (hence $p=3$). The $S^0$ sphere whose poset is nicknamed the diamond poset  can be easily identified by choosing any pair of vertices from layers $l_2$ and $l_4$ of every $(A_3,\id,s_1s_2s_3s_1s_2s_1,3)_{3}$ and checking whether they share zero or two vertices from $l_3$. No other option is possible~\cite{bjorner2006combinatorics}, see~\cref{prop:sphereClassification}~(1).

  A more interesting case for this work is the $S^1$ decomposition in terms of the $k$-crown poset. The $k$-crown posets are length three posets and so to identify a $k$-crown in layers $\{l_{i}\},i\in\{p-1,p,p+1\}$ one has to look for the corresponding bottom element $\hat{b}$ in  layer $l_{p-2}$ or for the top element $\hat{t}$ in layer $l_{p+2}$. One then obtains a left or a right $k$-crown poset. This decomposition is the main tool to get non-trivial CSS codes and~\cref{sec:crownSplicing} contains a detailed description.

  Finally, for the $S^2$ decomposition, that is, a poset of length four, one looks for the bottom/top element $\hat{b}/\hat{t}$ in layers $l_{p\mp2}$. I highlighted two of the identified $S^2$ face posets in~\cref{fig:A3S2decompositionTanner} together with the corresponding CW complexes  in~\cref{fig:A3S2decompositionCW}. A complete classification of the $S^2$ CW complexes for finite Weyl groups comes from~\cite{hultman2003combinatorial} and is listed in Fig.~\ref{fig:hultman}. I discuss in~\cref{sec:S0andS2decomposition} how to  use the $S^2$ to again create non-trivial CSS codes.

\end{exa}
\begin{figure}[t]
  \resizebox{8.5cm}{!}{\includegraphics{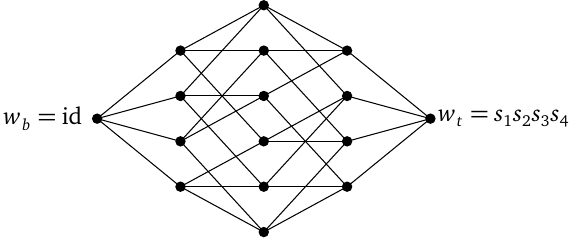}}
  \caption{The Bruhat/weak order of the reducible Coxeter group~$C_2^{\times4}$.  The open interval $(w_b,w_t)$ of length four is a face poset $\euF(S^2)$.}
  \label{fig:C2pow4}
\end{figure}

\begin{exa}
  Reducible group $C_2^{\times n}$: The cyclic group $C_2$  is a two-element Coxeter group with a simple structure ($(A_1,\{s_1\})$ in Coxeter's classification). Its Cayley graph is an edge and technically it is also a Bruhat face poset although it is an empty poset without any cells (faces). To get a non-trivial poset one takes $C_2\times C_2$ whose poset is the $S^0$ face poset: the diamond graph. It is formed by the Cartesian product (the same graph product used in the construction of the hypergraph product codes~\cite{tillich2013quantum}) of two edges resulting in a square and this is also the $C_2\times C_2$ Bruhat poset. So the Cayley and Bruhat orders coincide and this holds for all~$n$. What makes this group family potentially interesting is its relative simplicity, where the Cayley/Bruhat poset of $C_2^{\times n}$ is an $n$-dimensional cube also known as a Boolean lattice (a special kind of poset~\cite{davey2002introduction}), together with the fact that it leads to interesting codes as I will show in~\cref{sec:actI}. Case $n=4$ is depicted in~\cref{fig:C2pow4}.
\end{exa}

\begin{exa}
  As the final example I take an infinite-dimensional Coxeter group, namely one from an infinite family of the so-called hyperbolic triangle groups~\cite{davis2012geometry}:
  \begin{equation}\label{eq:triangleGroup}
    (\Delta_{\a,\b,\g},S)=\langle s_1,s_2,s_3|(s_1s_2)^\a,(s_1s_3)^\b,(s_2s_3)^\g\rangle,
  \end{equation}
  such that $1/\a+1/\b+1/\g<1$. They figure prominently in the QEC code design as codes created by compactifying hyperbolic tessellations. Let the triangle group be $\a=2,\b=3,\g=7$. I choose the standard $w_b=\id$ and $w_t=(s_1s_2s_3)^m$ in its reduced form. It is a locally finite group and so I can again depict its ranked Bruhat poset in the form of a Hasse diagram. Choosing $m=3$ the poset is depicted in~\cref{fig:237}. and unlike the previous two examples it is asymmetric. That is a typical poset Hasse diagram one gets from a randomly chosen Coxeter group unlike the highly symmetric groups $A_n$ or $C_2^{\times n}$ shown previously.
\end{exa}

\begin{figure}[t]
  \resizebox{14cm}{!}{\includegraphics{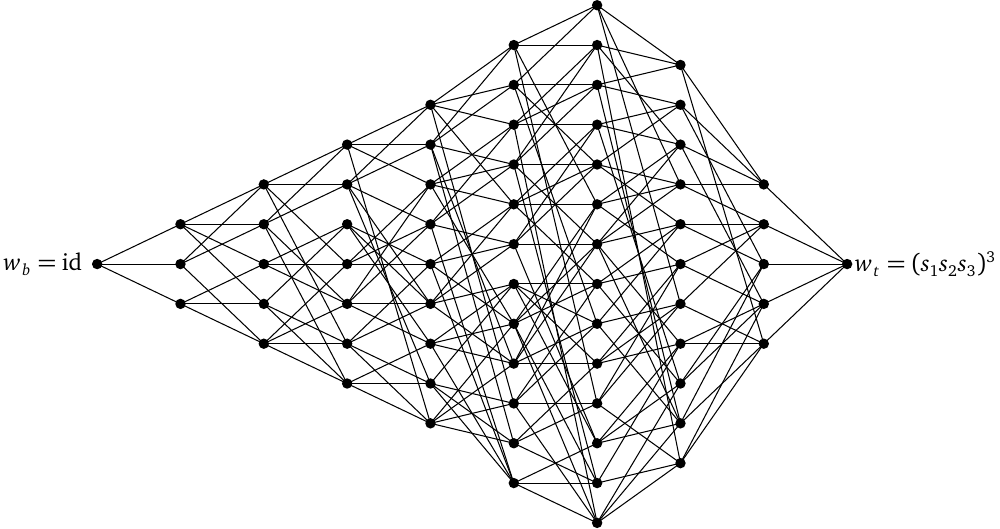}}
  \caption{A subposet of the  Bruhat order of the triangle hyperbolic group,~\cref{eq:triangleGroup}, for $\a=2,\b=3,\g=7$.  The open interval $(w_b,w_t)$ is a face poset $\euF(S^8)$.}
  \label{fig:237}
\end{figure}

\subsection{Trivial codes}

As a warm-up, the structure of the length two intervals in~\cref{prop:sphereClassification}~(1) in terms of the $S^0$ spheres and the corresponding diamond posets will be used as an alternative proof of the statement that every three-layer poset is a valid CSS code.  I assign a layer $l_p$ of $p$-cells to data qubits and the layers of $(p\pm1)$-cells to $X$ and $Z$ checks (or vice versa) for any consecutive triple of layers $\{l_{i}\},i\in\{p-1,p,p+1\}$  in the interval $[w_b,w_t]$. More formally:
\begin{lem}\label{lem:BruhatCodes}
    Let $(W,w_b,w_t,p)_{3}$ be a length two chain complex such that and $\ell([w_b,w_t])\geq5$. Further, let $\euQ \subseteq l_p$ be a poset layer of rank~$p$. Whenever $\euX\subseteq l_{p-1}$ and $\euZ\subseteq l_{p+1}$ are such that for all $x\in\euX$ and $z\in\euZ$, either $(x,z) \subseteq \euQ$ or $(x, z) \cap \euQ = \emptyset$ then the operators
    \begin{align*}
        \{\prod_{q \in \euQ: x\lessdot q} X_q : x \in \euX \}
    \end{align*}
    and
    \begin{align*}
        \{\prod_{q \in \euQ: q \lessdot z} Z_q : z \in \euZ \}
    \end{align*}
    define a CSS code where any $X$ and $Z$ check overlap on either zero or two qubits.
\end{lem}
\begin{proof}
Let $x \in\euX$ and $z \in \euZ$ correspond to $X$ and $Z$ checks that both act on some qubit $q\in\euQ$. Then by assumption, $(x,z) \subseteq \euQ$, $x < z$ and $\ell(z) - \ell(x) = 2$. Therefore $(x,z)$ contains two elements by \cite{bjorner2006combinatorics} (it is reproduced in Proposition~\ref{prop:sphereClassification}~(1)), where $[x,z]$ was dubbed the diamond poset in Appendix~\ref{app:sphereClassification}.
\end{proof}

Note that the codes of~\cref{lem:BruhatCodes} can encode logical qubits but their discovery relies on an intractable process of systematically generating subsets $\euX$ or $\euZ$ and checking the code properties. It is equivalent to a systematic removal of the rows of the corresponding PCMs. The diamond ($S^0$) decomposition of any $(W,w_b,w_t,p)_3$ CSS code sketched in~\cref{sec:S0andS2decomposition} makes it significantly more tractable but compared to the main decomposition method the resulting non-trivial codes are much less interesting and therefore not worth of pursuing.

The starting point of more efficient code constructions will be a special case of the previous lemma when $\euQ=l_p$ and $\euX=l_{p\pm1},\euZ=l_{p\mp1}$, where  $(W,w_b,w_t,p)_{3}$  uniquely determines a CSS code. It turns out, however, that every such code is trivial and the reason is that the underlying manifold, a $d$-dimensional sphere corresponding to the interval $(w_b,w_t)$, where $\ell((w_b,w_t))=d+2$, supports zero logical qubits for any~$d,W,p$ thanks to the known properties of the $S^d$ homology groups:
\begin{subequations}\label{eqs:SdHomology}
\begin{align}
     H_0(S^d,\bbZ_2) & = H_d(S^d,\bbZ_2) = \bbZ_2, \\
     H_i(S^d,\bbZ_2) & = 0\mbox{ otherwise}.
\end{align}
\end{subequations}

This perhaps intuitive fact is shown in the following lemma using the terminology of~\cref{appsub:CWhomology}:
\begin{lem}\label{lem:CSStriplesTriv}
  Let  $(W,w_b,w_t,p,\hat{b},\hat{t})_{3}$ be a length two chain complex  and $X$~a regular CW complex obtained by gradually removing the consecutive layers from its face poset $\euF(S^d)$. Then  $H_p(X,\bbZ_2)=0$.
\end{lem}
\begin{proof}
  All open intervals $(w_b,w_t)$ of length~$d+2$ pertaining to the Bruhat order of a Coxeter group  are graded face posets $\euF(S^d)$ of a regular CW complex homeomorphic to~$S^d$~\cite{bjorner2006combinatorics}. As such, the CW complex's Betti numbers are $\b_0=\b_d=1$ and zero otherwise, see~\eqref{eqs:SdHomology} and~\cref{appsub:CWhomology}. The boundary of the $(d+1)$-dimensional void are all $d$-dimensional cells. I remove from the face poset the layer ranked as $r=d$ corresponding to these $d$-cells and add a common top element $\hat{t}$ to be able to properly define all homology groups and the corresponding Betti numbers. In this way I obtain a new CW complex $X'$ whose dimension is reduced by one compared to the original sphere~$S^d$. The new complex is not, in general, homeomorphic to a sphere since by removing all the $d$-cells I effectively introduced a number of $d$-dimensional voids that are counted by the highest Betti number of~$X'$. Hence I reduced the $(d+1)$-tuple of Betti numbers ${\bbeta}(S^d)$ to a new $d$-tuple $\bbeta(X')$:
  \begin{equation}\label{eq:bettiReduction}
    {\bbeta}(S^d)=(1,0,\dots,0,1)\to\bbeta(X')=(1,0,\dots,0,\b_{d-1}(X')),
  \end{equation}
  where $\b_{d-1}(X')=\rank{\ker{\partial_{d-1}}}$ as follows from~\eqref{eq:rankHp}. The rest of the Betti numbers remains the same since I haven't touched the other boundary operators.  I continue by removing layer by layer, thus erasing the previously introduced voids and introducing the voids of one dimension less, again counted by the corresponding Betti number. In the penultimate step where only three layers of the Hasse diagram depicting the face poset are left I reduced $S^d$ to a manifold $X''$ whose Betti numbers are
  \begin{equation}\label{eq:bettiPenultim}
    {\bbeta}(X'')=(1,0,\b_2(X'')).
  \end{equation}
  corresponding to the chain complex
  \begin{equation}\label{eq:chainComplexSubposet1}
    \hat{t}\to C_2\to C_1\to C_0\to\hat{w}_b.
  \end{equation}

  I now generalize this process to any neighboring triple of layers. Let the desired triple of layers from $\euF(S^d)$, without loss of generality, be $l_i,i\in \{p-1,p,p+1\}$. Following the same procedure described in the previous paragraph I stop upon reaching a CW complex that I now call~$X'$, whose Betti $(q+1)$-tuple reads
  \begin{equation}\label{eq:bettiFromTop}
    \bbeta(X')=(1,0,\dots,0,\b_{q}(X'))
  \end{equation}
  for $q\leq d$.   I now reinterpret the face poset of $X'$ by considering its $q$-cells to become the $(d-q)$-cells of $\widetilde{X}'$. Indeed, there is a complete freedom from which end (bottom $w_b$ or top $w_t$) of $\euF(S^d)$ I start assigning 0-cells. I effectively dualized $X'$ and with this transformation in hand I may again start removing the highest cells, layer by layer, but this time from  $\widetilde{X}'$ (that is, from the other side of the Hasse diagram of the original $X'$). In the penultimate step I arrive at
  \begin{equation}\label{eq:bettiFromBotTop}
    \bbeta(\widetilde{X}')=(\b_{p-1}(\widetilde{X}'),0,\b_{p+1}(\widetilde{X}')),
  \end{equation}
  where $\b_{p+1}(\widetilde{X}')\equiv\b_{q}(X')$ from Eq.~\eqref{eq:bettiFromTop}. Graphically, I obtained the chain complex
    \begin{equation}\label{eq:chainComplexSubposet2}
    \hat{t}\to C_{p+1}\to C_p\to C_{p-1}\to\hat{b},
  \end{equation}
  which I can identify  with   $(W,w_b,w_t,p,\hat{b},\hat{t})_{3}$ from~\cref{def:BruhatChains}. Since $\b_{p}(\widetilde{X}')=0$  the claim follows.
\end{proof}
\begin{cor}\label{cor:CSStriplesTriv}
  The closed interval  $[\hat{b},\hat{t}]=(W,w_b,w_t,p,\hat{b},\hat{t})_{3}$ (chain complex of length two) corresponds to the open poset interval $(W,w_b,w_t,p)_{3}$ with three layer (see~\cref{def:BruhatChains}) that can be interpreted as a Tanner graph of a CSS code. Thanks to the $p$-th homology group being zero it encodes zero logical qubits.
\end{cor}

This result basically shows that `slicing' a trivial manifold along any triple of consecutive layers of its CW face poset still yields a topologically trivial structure. This looks like a showstopper for the purpose of generating interesting QEC codes. Indeed, the trivial  $(W,w_b,w_t,p)_{3}$ CSS codes are by construction topological since all vertices of the resulting Tanner graph were identified with closed cells of a certain dimension. One direction could be to seek quotient maps just like it was done for the spaces of constant curvature. But this has  led to a large variety of codes only in the case of hyperbolic spaces~\cite{guth2014quantum,breuckmann2016constructions,breuckmann2021single}, whereas for the spherical and planar geometry it `merely' provides codes on the projective plane $\bbR\bbP^n$ (such as the $[9,1,3]$ Shor code for $n=2$) and the generalized toric codes. Given the proximity of the  $(W,w_b,w_t,p)_{3}$ codes to the spherical geometry this strategy does not look promising. Instead, let's use the particular structure of the high-dimensional spheres obtained from the Bruhat order on Coxeter groups. This structure is given by the sphere decomposition of the length two, three and four intervals described in Appendix~\ref{app:sphereClassification}.

\subsection{Non-trivial codes}\label{sec:nontrivCSS}

My first strategy to modify a  $(W,w_b,w_t,p)_{3}$ CSS code  to encode logical qubits will be with the help of~\cref{prop:sphereClassification} and one of the three ways every 3-layer poset (a CW chain complex of length two or a Tanner graph)  given by the Bruhat order  can be decomposed into elementary building blocks by: (1) the diamond posets, (2) $k$-crown posets and, (3), the posets for Hultman's 24 $S^2$ spheres~\cite{hultman2003combinatorial}, see~\cref{fig:hultman}, if the used Coxeter group is of a Weyl type. If it is not a Weyl group, it is still not an issue since any new sphere can be easily identified and added to the list, see~\cref{app:sphereClassification} for more detail. When I say easily, it will be clear that once a face poset of $(W,w_b,w_t,p)_{2k+1}$ is derived (always needed in the following analysis) there is no computational complexity hiding in finding the corresponding $S^0, S^1$ or $S^2$ sphere decompositions.  It is less clear, however, what the complexity is to actually get the $(2k+1)$-level poset or even the entire initial  poset $\euF(S^d)$ for an interval $[w_b,w_t]$, Coxeter group and its order. Indeed, this problem goes deep and among other things it touches the complexity of the rewrite rules necessary to find the reduced length of a word  for a general finitely presented infinite discrete groups such as infinite Coxeter groups. The known results are encouraging~\cite{brink1993finiteness} but this is a topic for a separate investigation. Even though the Coxeter groups studied here are small enough for us not to be bothered by these fundamental issues, they will provide quite large codes going beyond what one would call  toy CSS code examples.

\subsubsection{$S^1$ (crown) decomposition}\label{sec:crownSplicing}

I first define an operation I call `stabilizer splicing'. For the notation see~\cref{app:homoCSS}.
\begin{defi}\label{def:splice}
  Let $H_X,H_Z$ be PCMs of an arbitrary CSS code whose rows are denoted $h^i_X,h^j_Z$. A spliced $X$ stabilizer $\tilde H_X$ is obtained from $H_X$ by choosing  $I\subseteq [\dim{H_X}]$ to generate a new (linearly-dependent) $X$ stabilizer
  \begin{equation}\label{eq:splice}
    \tilde{h}^k_X=\sum_{i\in I}h^i_X \mod 2,
  \end{equation}
  such that $ \tilde{h}^k_X$ substitutes $h^i_X$ for all $i\in I$. Similarly for the spliced $Z$ stabilizer $\tilde{H}_Z$.
\end{defi}
Spliced stabilizers are valid PCMs. That is, $\tilde H_X \tilde H_Z^T=0$ holds thanks to the linearity of the splicing operation (element-wise sum modulo two). This and related constructions will be the main tool in transforming  trivial Bruhat-based CSS codes to non-trivial ones. But what can be expected from the spliced CSS codes? At first sight it doesn't look encouraging. I present two of its negative properties. First, there is no obvious guarantee that the resulting CSS code is protecting all or even any logical/physical qubit. Consider an overprotected $[4,2,2]$ code with two $X$ and two $Z$ checks coupled to all four data qubits. Splicing the $X$ or $Z$ stabilizers will leave the data qubits unprotected, see~\cref{fig:422}. Second, the splicing changes the stabilizers' weights. Some qubit connections  (that is, the edges of a Tanner graph) disappear and, in general, some of the new spliced stabilizers become heavier but in a highly non-uniform way shown later in~\cref{fig:histogram}.
\begin{figure}[t]
  \resizebox{7.5cm}{!}{\includegraphics{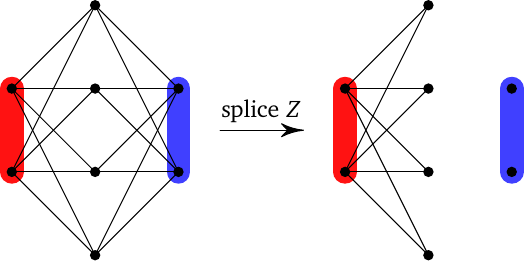}}
  \caption{(Left) A Tanner graph of the $[4,2,2]$ code with redundantly protected qubits to illustrate the negative  effect of splicing. (Right) $Z$ stabilizer  (indicated in blue) splicing exposes the data qubits to $X$ errors.}
  \label{fig:422}
\end{figure}

Despite these issues  splicing turns out to be quite useful in general and mainly in connection with the $S^1$ decomposition of the $(W,w_b,w_t,p)_3$ codes. I will describe an algorithm whose main input is the complete $k$-crown structure of  a chosen code  $(W,w_b,w_t,p)_3$. $k$-crowns are not unknown structures to QEC practitioners. Although it is a length three poset due to the presence of the bottom and top vertices (see $\hat{b},\hat{t}$ in~\cref{fig:S012} for $k=7$), once both are removed one is left with a classical Tanner graph of a (redundant) $k$-repetition code. Similarly, if either $\hat{b}$ or $\hat{t}$ is removed  a quantum Tanner graph of a (trivial) quantum CSS code with redundant stabilizers of one type is obtained.

Given a $(W,w_b,w_t,p)_{3}$ CSS code I will introduce the left and right $k$-crown graphs. To this end, consider a longer chain complex $(W,w_b,w_t,p)_{5}$ consisting of five consecutive layers $l_{p-2},l_{p-1},l_{p},l_{p+1}$ and $l_{p+2}$ of the Bruhat face poset $\euF(S^d)$ for a sufficiently big~$d$. Then according to~\cref{prop:sphereClassification}, item~(2), for each pair of elements $\hat{b}\in l_{p-2},\hat{t}\in l_{p+1}$ there exists a $k$-crown poset with this bottom and top element. Such poset will be called a \emph{left $k_L$-crown} poset. Similarly,  for each pair of elements $\hat{b}\in l_{p-1},\hat{t}\in l_{p+2}$ there exists a $k_R$-crown poset equipped with this bottom and top element. Any such $S^1$ sphere will be called a \emph{right $k_R$-crown} poset. I find all left and right $k_{L/R}$-crowns by iterating over $|l_{p-1}|\times|l_{p+2}|+|l_{p-2}|\times|l_{p+1}|$ possibilities and consider their subposets obtained by removing $w_b$ from all left $k_L$-crowns and $w_t$ from all $k_R$-right crowns. The collection of these subposets will be denoted $\{\euF_L(S^1)\}$ and $\{\euF_R(S^1)\}$ --~\cref{fig:S1LeftRight} illustrates the construction.

\begin{figure}[t]
  \resizebox{8cm}{!}{\includegraphics{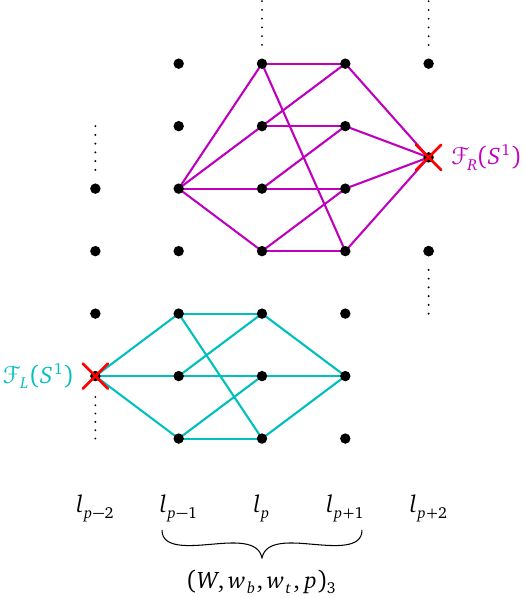}}
  \caption{Five poset layers $l_i,p-2\leq i\leq p+2$ of a general Bruhat poset $\euF(S^d)$ used to find all left and right $k_{L/R}$-crown posets of a CSS code $(W,w_b,w_t,p)_3$. The left $k_L$-crown is illustrated for $k_L=3$ in cyan and the right $k_R$-crown for $k_R=4$ in magenta. The terminal vertices ($\hat{b}$ for the left crown and~$\hat{t}$ for the right crown) and their incident edges that do not participate in the CSS code  are removed (indicated by the red crosses).}
  \label{fig:S1LeftRight}
\end{figure}

\cref{algo:crwonSplicing} is a pseudocode describing the process of obtaining a non-trivial CSS code from a $(W,w_b,w_t,p)_{3}$ CSS code. In words, I randomly choose a total number $\kappa$ of left or right $k$-crowns with probability~$\wp$\footnote{The use of the left ($L$) and right ($R$) notation reflects the vertical arrangement of the face posets/Tanner graphs used in this paper. It then lets us decide whether $L=X(Z)$ and $R=Z(X)$.}. This is a bias introduced to take into account the asymmetric Tanner graphs of $(W,w_b,w_t,p)_{3}$. A typical Bruhat order-based code has a different number of $X$ or $Z$ checks and an even splicing would lead to suboptimal codes. The bias parameter can be calculated in different ways. I use
$$
\wp={|\{\euF_L(S^1)\}|\over|\{\euF_L(S^1)\}|+|\{\euF_R(S^1)\}|}
$$
but another option could be
$$
\wp={|l_{p-1}|\over|l_{p-1}|+|l_{p+1}|},
$$
where $|l_{p\pm1}|$ is the number of left/right check qubits. Another important parameter is the overlap~$\la$. When a new $k$-crown is picked, the corresponding left/right checks can coincide with the previously chosen left/right checks in at most $\la$ instances. This turns out to be quite an important parameter, where the best codes are found for $\la=1$. Some choices of $\kappa,\la$ may not be (or are unlikely to be) satisfied and so a cut-off parameter~$c$ is introduced to avoid getting stuck in a loop. The subprocedure \textsc{Splice} is the actual splicing as described in~\cref{def:splice}.

\begin{algorithm}[t]
	\caption{$k$-crown ($S^1$) splicing of a $(W,w_b,w_t,p)_{3}$ CSS code}
	\label{algo:crwonSplicing}
	\begin{algorithmic}[1]
		\State {\bfseries Input}:
		\State PCMs $H_L,H_R$ of a  $(W,w_b,w_t,p)_{3}$ CSS code.
		\State All left and right $k$-crown subposets $\{\euF_L(S^1)\},\{\euF_R(S^1)\}$  of the $(W,w_b,w_t,p)_{3}$ CSS code.
        \State Total number $\kappa$ of $k$-crowns to be spliced.
		\State Overlap parameter $\la$, cut-off parameter $c$ and bias probability~$\wp$.
		\State {\bfseries Output}:
		\State PCMs $\tilde H_L,\tilde H_R$ of the spliced  $(W,w_b,w_t,p)_{3}$ CSS code.	
        \Procedure {CrownSplice}{$H_L,H_R,\{\euF_L(S^1)\},\{\euF_R(S^1)\},\kappa,\la,\wp,c$}
        \State $n\gets 0,I\gets\emptyset,J\gets\emptyset$
	 	\For{$m \gets 1$ to $\kappa$}
        \Repeat
        \State Flip a biased coin of probability $\wp$ to pick an element of $\{\euF_L(S^1)\}$ or $\{\euF_R(S^1)\}$
        \State Find the PCMs' rows $\{h_L^m\}$ or $\{h_R^m\}$ to be spliced
        \State ($I_m\gets \mathrm{indrow}{[\{h_L^m\}]}$  {\bfseries and } $J_m\gets\emptyset$)  {\bfseries or } ($I_m\gets\emptyset$  {\bfseries and } $J_m\gets \mathrm{indrow}{[\{h_R^m\}]}$) \\\Comment\#The function $\mathrm{indrow}$ identifies the row indices
        \State $n\gets n+1$
        \Until{$|I\cap I_m|\leq\la$} {\bfseries or }{$|J\cap J_m|\leq\la$} {\bfseries or }$n>c$
        \State $I\gets\cup_m I_m$
        \State $J\gets\cup_m J_m$
        \EndFor

        \Procedure{Splice}{$H_L$,$H_R$,$\{I_m\}$,$\{J_m\}$}
        \State $\tilde H_L\gets H_L$
        \State $\tilde H_R\gets H_R$
	 	\EndProcedure
        \EndProcedure
	\end{algorithmic}
\end{algorithm}

Once $\tilde H_L,\tilde{H}_R$ are obtained, the data qubits that became decoupled from one or both checks are removed. This doesn't spoil the commutation relations. The distance parameter $d=\min{[d_L,d_R]}$ is then estimated (upper-bounded) and exactly verified by dist-m4ri\textsuperscript{\ref{fn:distm4ri}} or Qubitserf\textsuperscript{\ref{fn:qubitserf}} if time or memory permits. The choice of a triple of input parameters $\kappa,\la,c$ has to be optimized but it is fairly easy to guess a critical $\kappa$ value, where the number of logical qubits becomes non-zero. Note that one can observe that the stabilizers of the  $(W,w_b,w_t,p)_{3}$ codes are highly overcomplete, that is, $\dim{H_{X/Z}}>\rank{H_{X/Z}}$ and from the difference it is straightforward to deduce an interesting $\kappa$ regime. In fact, the behavior of ~\cref{algo:crwonSplicing} is such that the switch from trivial to non-trivial codes is not a gradual transition as a function of~$\kappa$. So there is no need to sweep large swaths of the parameter space. \cref{algo:crwonSplicing} produces just a single code instance and as I mentioned previously there is no guarantee that the code is interesting. It is therefore repeated with the same input parameters but even the largest codes presented in the next section were obtained without taking an excessive number of samples. The main factor preventing me from going to even larger trivial codes as a starting point is not the time complexity of taking many samples but rather the need to estimate the distance each time. All tools I use are either sampling algorithms themselves~\cite{QDistRnd} or their complexity scales badly\textsuperscript{\ref{fn:distm4ri},\ref{fn:qubitserf}} (as a matter of inevitability) and this prevents from going further without using additional parallel computational resources.

I, nevertheless, introduce a random proxy method to crown splicing in~\cref{sec:otherCSSconstructions}. In this case, I completely disregard the code's internal structure and simply splice a random number of random pairs of $X$ or $Z$ Tanner graph vertices (could be generalized to random $n$-tuples). Despite taking a small number of random samples the method is remarkably stable and provides almost identical codes with high probability. How good are these codes? Random splicing applied to $(W,w_b,w_t,p)_{3}$ CSS codes yields (most of the time) sub par codes to crown splicing in terms of both the distance and mainly the stabilizer weight. But it is quite informative as what to expect should the crown splicing procedure be performed. Additionally, it points to an even more exciting possibility that the crown posets do indeed play a role in the internal structure of the Bruhat codes derived from Coxeter groups. I leave this as a major open question.

Before showing some of the found codes, I mention in passing that the sets $\{\euF_L(S^1)\}$ and $\{\euF_R(S^1)\}$ contains $k$-crowns for various $k$'s and they correspond to $2k$-cycles as depicted in~\cref{app:sphereClassification}. Among them is the smallest one, $k=2$, and the corresponding 4-cycle. They exactly correspond to the girth of the Tanner graph and serve  as a measure of how certain decoding algorithms perform~\cite{roffe2020decoding}. It is not clear if the short cycles can be entirely avoided by exploring or designing Coxeter groups. Moreover, the splicing procedures, in general, alter the cycle structure of the codes.

\subsubsection{The case of $S^0$ and $S^2$ decompositions}\label{sec:S0andS2decomposition}

The crown decomposition hit the sweet spot when it comes to the derivation of non-trivial CSS codes. But as one can see in~\cref{app:sphereClassification}, all $k$-crown graphs are made of diamonds (the $S^0$~diamond posets) and all $S^2$ posets are made of $k$-crowns. Can they be used for non-trivial code creation? Besides splicing, there exists another strategy to obtain non-trivial codes suitable to the diamond structure. It is simply by the systematic removal of rows from the $X$ or $Z$ stabilizer. However, without any internal structure, this procedure is hopelessly inefficient due to the combinatorial explosion in the number of possibilities of the checks to remove. This is not the strategy I am going to pursue here but  it is worth mentioning that the internal structure of the trivial codes in terms of their diamond posets is indeed useful.
If I remove pairs of $X,Z$ checks forming a diamond non-trivial codes are obtained again and in a significantly more efficient way than just randomly removing stabilizer rows. Indeed, once all the diamonds are tracked (there are at most $|l_{p-1}|\times|l_{p+1}|$ of them) they provide the desired pairs of checks. The upside is that the resulting CSS codes do not suffer from the stabilizer weight gain like for $k$-crown splicing. The downside is that the code parameters such as the distance and mainly the rate are worse, see~\cref{table:S1S2splicing}.

How does the crown splicing compare to the $S^2$ splicing? Again, using the length four chain complex like in the crown decomposition I can find all $S^2$ posets as classified in~\cite{hultman2003combinatorial} for Weyl groups and depicted in~\cref{fig:hultman} or even, in general, new ones for non-Weyl groups. With this input, a simple modification of~\cref{algo:crwonSplicing} is used. Unlike the left and right $k$-crowns, where one of the check types is always a single vertex (so nothing to splice), every $S^2$ poset  contains more than one check of both the  $X$ and $Z$ type. They are then spliced  (each type separately). However, as explained in~\cref{app:sphereClassification}, $k$-crowns are the cycles visible in~\cref{fig:hultman} depicting the CW complexes of $S^2$ spheres. One may guess that thanks to the `correlated' $X$ and $Z$ splicing the $S^2$ cannot generate better codes than crown splicing. Indeed, given a code obtained from the $S^2$ splicing one can formally transform it into a code obtained from some crown splicing but not vice versa. Therefore, crown splicing is a richer procedure to get non-trivial codes. However, the $k$-crowns in a given CW $S^2$ complex share edges (typically more than one) and this corresponds to greater crown overlapping than I used in~\cref{algo:crwonSplicing} by choosing $\la=1$. So this type of splicing can provide new codes typically not found by crown splicing and I indeed did find one example, see~\cref{table:S1S2splicing}. Another advantage of the $S^2$ splicing is the speed due to a smaller number of the $S^2$ posets compared to the left and right $k$-crowns. There is another, speculative, option: that  the structure of the $S^2$ posets has something to do with the existence of interesting codes and their splicing should be preferred. I leave it as another open question.

\subsubsection{Other CSS code constructions -- random splicing and chain complex folding}\label{sec:otherCSSconstructions}

In~\cref{sec:crownSplicing} I showed how to  decompose any $(W,w_b,w_t,p)_{3}$ CSS code in terms of the $k$-crown posets. The results obtained by splicing based on this approach are quite interesting codes -- at least when it comes to the number of encoded qubits and distance. The heuristic code discovery adventure could end here but are there other ways of how to make the  $(W,w_b,w_t,p)_{3}$ CSS codes non-trivial (and competitive in terms of the code parameters)? In this section I sacrifice the detailed structure of these codes in favor of a more, one could say, blunt but at the same time more efficient approach of how to get interesting codes. The introduced procedure is not only more computationally feasible compared to crown splicing but also allows me to extend the scope of where the Bruhat order of Coxeter group can be applied. Namely, by employing the chain complexes of lengths greater than two I describe a method of how to obtain CSS codes with metachecks of one type, which were shown to be useful for single-shot decoding~\cite{bombin2015single,campbell2019theory,breuckmann2021single,quintavalle2021single,higgott2023improved}.

In terms of chain complexes the $S^1$ or $S^2$ splicing procedure  can be summarized  as
\begin{equation}\label{eq:chacoSplice}
  C_{p+1}\overset{H_Z^T}\to C_p\overset{H_X}\to C_{p-1}\overset{\mathrm{splice}}\longmapsto
  \tilde C_{p+1}\overset{\tilde H_Z^T}\to\tilde C_p\overset{\tilde H_X}\to\tilde C_{p-1}.
\end{equation}
The splicing alters the $(p\pm1)$-cell structure indicated by tildes on $C_{p\pm1}$. But thanks to an occasional decoupling of the data qubits, which are then removed, I put a tilde on $C_p$ as well. That does not mean the data qubits have been spliced -- it would violate the CSS condition. But what if the sphere substructure of $H_X$ and $H_Z$ doesn't really matter? To test it, I assume $\dim{H_X}=2c_x,\dim{H_Z}=2c_z,c_i\in\bbZ^+$ and I randomly match pairs of stabilizers to be spliced -- independently for $X$ and $Z$. The number of matchings is $(c_x!!)\times(c_z!!)$, where $!!$ is the double factorial. The action of the resulting PCMs $\tilde H_X,\tilde H_Z$, where $\dim{\tilde H_X}=c_x,\dim{\tilde H_Z}=c_z$, as linear maps between the modules is captured by the same chain complex as on the RHS of~\cref{eq:chacoSplice}. Clearly, it would be pointless to generate all possible matchings but for all tested  $(W,w_b,w_t,p)_{3}$ CSS codes I consistently obtain decent non-trivial codes with high probability, see~\cref{sec:actI} for a more nuanced discussion.

The most interesting thing about this kind or random splicing is that it seems to behave like a good proxy to the  computationally more demanding $k$-crown splicing described in~\cref{algo:crwonSplicing}  in terms of the code parameters. Whenever I could compare it, such crudely spliced codes are often  worse both in the logical qubit space dimension $2^k$ and distance but not much. Whenever I am not able to verify it due to the slow speed of the repeated distance calculations I conjecture that the random splicing can indeed be considered as a tight lower bound on $k$-crown splicing code parameters. There are, however, two distinguishing features, where the crown spliced codes significantly outperform the randomly spliced ones: the stabilizer weight and the code rate within prospective code families. As again discussed in~\cref{sec:actI}, the crown spliced stabilizer weight distribution is highly skewed whereas the randomly spliced codes are quite uniform and on average much heavier. For the code rate, unlike crown splicing where the rate seems to be improving, it is decreasing within a code family.

For now I will take the suboptimal random splicing construction as a tool to quickly generate new classes of CSS codes, this time from longer chain complexes $(W,w_b,w_t,p)_{2k+1}$ for $k=2,3$\footnote{Note that I can in principle substitute the fast random splicing by the slower crown splicing (the sphere substructure is still present) and presumably get even better results. This is nevertheless computationally demanding.}. To this end, I define a procedure dubbed \emph{chain complex folding}.
\begin{defi}\label{def:folding}
  Let $C_\bullet=(C_i,K_i)$ be a length $2k$ chain complex
  \begin{equation}\label{eq:Cbullet}
    0\to C_{p+k}\overset{K_{p+k}}\longrightarrow C_{p+k-1}\ \cdots\ C_{p-k+1}\overset{K_{p-k+1}}\longrightarrow  C_{p-k}\to 0,
  \end{equation}
  where $\{C_i(X,\bbZ_2)\},p-k\leq i\leq p+k$ are finite-dimensional modules and $K_i$ linear boundary maps. A new chain complex called \emph{folded chain complex} is formed by reversing the first $k$ arrows of $C_\bullet$ (cf.~\cref{eq:coChaincomplex}):
  \begin{equation}\label{eq:Cfolded}
     0\to C_{p}{\ \looongrightarrow{K_p\boxplus K_{p+1}^T}} \ C_{p-1}\oplus C_{p+1}\ {\looongrightarrow{K_{p-1}\oplus K_{p+2}^T}}\ C_{p-2}\oplus C_{p+2}\cdots C_{p-k+1}\oplus C_{p+k-1}\ {\looongrightarrow{K_{p-k+1}\oplus K_{p+k}^T}}\ C_{p-k}\oplus C_{p+k}\to0.
  \end{equation}
   Furthermore, one may decide to apply a similar operation to folding on the last map:
    \begin{equation}\label{eq:CfoldedTwice}
     0\to C_{p}{\ \looongrightarrow{K_p\boxplus K_{p+1}^T}} \ C_{p-1}\oplus C_{p+1}\ {\looongrightarrow{K_{p-1}\oplus K_{p+2}^T}}\ C_{p-2}\oplus C_{p+2}\cdots C_{p-k+1}\oplus C_{p+k-1}\ {\looongrightarrow{(K_{p-k+1}\boxplus K_{p+k}^T)^T}}\ C_{p-k,p+k}\to0,
  \end{equation}
  where the codomain of $(K_{p-k+1}\boxplus K_{p+k}^T)^T$ is denoted $C_{p-k,p+k}$.
\end{defi}

The initial and terminal zeros in both sequences are zero modules. The symbol $\oplus$ for the modules stands for an (internal) direct sum. The same symbol for the linear map $K_i$ is a direct sum of their matrix representations. The symbol $\boxplus$ stands for vertical concatenation of the linear maps $K_i$ represented as matrices. Specifically~\cite{bhatia2013matrix}, let $K:V\mapsto W$, where $V=V_1\oplus V_2,W=W_1\oplus W_2$ are internal direct sums of modules or vector spaces $V_i,W_j$ and $K\in\L(V,W)$ -- the space of linear maps. Then, $K$ is of the block matrix form
$$
K=\begin{bmatrix}
   K_{11} & K_{12} \\
   K_{21} & K_{22} \\
 \end{bmatrix},
$$
where $K_{ji}:V_i\mapsto W_j$. So when I write, for example, $K_{11}\boxplus K_{21}$ I mean the first column of $K$ whereas $K_{11}\oplus K_{22}$ is the diagonal part of~$K$\footnote{The notion of internal/external direct sum of vector spaces/modules and a direct sum of matrices is a source of confusion and that's why I want to be clear here.}.

\begin{rem}
  Next lemmas show explicitly why the folded structure is a valid chain complex for  $C_\bullet=(W,w_b,w_t,p,\hat{b},\hat{t})_{2k+1}$ but it holds in general by realizing that folding is identical to splicing -- if I consider $K_p$ and $K_{p+1}^T$ to be stabilizers (I can do that even though the folded complex is of length more than two) the support of the spliced stabilizers is two disjoint sets of data qubits. So the splicing becomes the (transposed) $\boxplus$ operation mapping from $C_p$ to $C_{p-1}\oplus C_{p+1}$. Since I know that splicing produces valid CSS codes I therefore get a valid (folded) chain complex.

  It also implies that the weight of $K_p\boxplus K_{p+1}^T$ considered to be a stabilizer is a sum of the weights of $K_p$ and $K_{p+1}$.
\end{rem}

%\begin{equation}\label{eq:5chain}
%\begin{tikzcd}[column sep=large,row sep=large]
%C_{p+1} \arrow{r}{H_Z^T}  & C_p \arrow{r}{H_X}  \arrow[mapsto]{d}{\textsc{Splice}} & C_{p-1}\\
%C & A
%\end{tikzcd}
%\end{equation}

%\begin{equation}\label{eq:}
%\begin{tikzcd}
%C_{p+1} \arrow{r}{H_Z^T}  & C_p \arrow{r}{H_X}  \arrow[swap]{d}{\varrho_x^f} & C_{p-1} \overset{\textsc{Splice}}\longmapsto \\
%C & A
%\end{tikzcd}
%\end{equation}

Consider the following transformations of a  length four chain complex  $(W,w_b,w_t,p)_{5}$~(\cref{def:BruhatChains}) and omitting the zero modules:
\begin{subequations}\label{eq:5chain}
\begin{align}\label{eqs:5chainA}
  C_{p+2}\overset{K_{p+2}}\longrightarrow C_{p+1}\overset{K_{p+1}}\longrightarrow &\ C_p\overset{K_p}\longrightarrow  C_{p-1}\overset{K_{p-1}}\longrightarrow  C_{p-2}\\
  &\ \ \rotatebox[origin=c]{-90}{$\longmapsto$}\mbox{\ folding}\nonumber\\\label{eqs:5chainB}
  C_{p}{\ \looongrightarrow{K_p\boxplus K_{p+1}^T}} \ C_{p-1}&\oplus C_{p+1}\ {\looongrightarrow{(K_{p-1}^T\boxplus K_{p+2})^T}}\ C_{p-2,p+2}\\
  &\ \ \rotatebox[origin=c]{-90}{$\longmapsto$}\mbox{\ splicing}\nonumber\\\label{eqs:5chainC}
  \tilde C_{p}{\ \looongrightarrow{\tilde K_p\boxplus\tilde K_{p+1}^T}} \ \tilde C_{p-1}&\oplus \tilde C_{p+1}\ {\looongrightarrow{(\tilde K_{p-1}\boxplus \tilde K_{p+2}^T)^T}}\ \tilde C_{p-2,p+2}.
\end{align}
\end{subequations}
Then, the following holds:
\begin{lem}\label{lem:5chain}
  Sequences~\cref{eqs:5chainB} and~\eqref{eqs:5chainC} are length two chain complexes and therefore CSS codes for any length four chain complex  $(W,w_b,w_t,p)_{5}$ in~\cref{eqs:5chainA}.
\end{lem}

\begin{proof}
  Since \cref{eqs:5chainA} is a chain complex of length four it gives rise to shorter chain subcomplexes of length two. In particular,  codes $(W,w_b,w_t,p)_{3}, (W,w_b,w_t,p-1)_{3}$ and  $(W,w_b,w_t,p+1)_{3}$ are trivial (due to~\cref{cor:CSStriplesTriv}) CSS codes. For the last two supported by the three leftmost and the three rightmost modules, one obtains (cf.~\cref{eq:chacoSplice})
  \begin{align}%\label{eqs:}
    K_{p-1}K_p & =0,\\
    K_{p+1}K_{p+2} & =0
  \end{align}
  holds. But then the map concatenation in~\cref{eqs:5chainB} satisfies
  \begin{equation}\label{eq:CSS3chain2fold}
    (K_{p-1}^T\boxplus K_{p+2})^T(K_p\boxplus K_{p+1}^T)=K_{p-1}K_p+K_{p+2}^TK_{p+1}^T=0
  \end{equation}
  and so it is a valid length two chain complex (CSS code). If only single-sided folding is used following~\cref{eq:Cfolded} I again get
  \begin{equation}\label{eq:CSS3chain}
    (K_{p-1}\oplus K_{p+2}^T)(K_p\boxplus K_{p+1}^T)=K_{p-1}K_p+K_{p+2}^TK_{p+1}^T=0.
  \end{equation}
  finally,~\cref{eqs:5chainC} is obtained by splicing as performed on the RHS of~\cref{eq:chacoSplice} again leading to a valid CSS code.
\end{proof}
Whereas none of the length two chains in~\cref{eqs:5chainA} gives a non-trivial code  the same is not true for~\cref{eqs:5chainB} and~\eqref{eqs:5chainC}. The splicing operation resulting in~\cref{eqs:5chainC} (and later in~\cref{eqs:7chainC}) is an added option to increase the rate of the folded codes but at the expense of the stabilizer weight increase. I will not take advantage of it since I demonstrated random splicing enough. I will, however, show examples of non-trivial folded codes in~\cref{table:folding}.
\begin{rem}
  One cannot not notice the similarity between the way generalized bicycle codes are constructed. Indeed, the box sum in~\cref{eqs:5chainB} is a vertical concatenation of matrices and so $H_Z^T\equiv K_p\boxplus K_{p+1}^T$ in~\eqref{eq:chacoSplice} resembles the construction from~\cite{kovalev2013quantum,bravyi2024high}. But that's where the similarity ends since neither the linear maps $K_i$'s forming $H_X\equiv (K_{p-1}^T\boxplus K_{p+2})^T$ are  directly related to those from $H_Z^T$ nor they form a commutative ring (of anything obvious) like the cyclic shifts in bicycle codes.
\end{rem}

The procedure in~\cref{lem:5chain} can be pushed one step further to generate non-trivial CSS codes with a metacheck corresponding to length three chain complexes. To that end, consider an arbitrary chain complex $(W,w_b,w_t,p)_{7}$ of length six in~\cref{eqs:7chainA}:
\begin{subequations}\label{eq:7chain}
\begin{align}\label{eqs:7chainA}
  C_{p+3}\overset{K_{p+3}}\longrightarrow C_{p+2}\overset{K_{p+2}}\longrightarrow C_{p+1}\overset{K_{p+1}}\longrightarrow &\ C_p\overset{K_p}\longrightarrow  C_{p-1}\overset{K_{p-1}}\longrightarrow  C_{p-2}  {\looongrightarrow{K_{p-2}}}\ C_{p-3}\\
  &\ \ \rotatebox[origin=c]{-90}{$\longmapsto$}\mbox{\ folding}\nonumber\\\label{eqs:7chainB}
  C_{p}{\ \looongrightarrow{K_p\boxplus K_{p+1}^T}} \ C_{p-1}\oplus C_{p+1}\ {\looongrightarrow{K_{p-1}\oplus K_{p+2}^T}}&\ C_{p-2}\oplus C_{p+2}
    \ {\looongrightarrow{(K_{p-2}^T\boxplus K_{p+3})^T}}\ C_{p-3,p+3}\\
  &\ \ \rotatebox[origin=c]{-90}{$\longmapsto$}\mbox{\ splicing}\nonumber\\\label{eqs:7chainC}
  \tilde C_{p}{\ \looongrightarrow{\tilde K_p\boxplus\tilde K_{p+1}^T}} \ \tilde C_{p-1}\oplus\tilde C_{p+1}\ {\looongrightarrow{K_{p-1}\oplus K_{p+2}^T}}&\ \tilde C_{p-2}\oplus\tilde C_{p+2}
  \ {\looongrightarrow{(\tilde K_{p-2}^T\boxplus\tilde K_{p+3})^T}}\ \tilde C_{p-3,p+3}.
\end{align}
\end{subequations}
\begin{lem}
  Sequences~\cref{eqs:7chainB} and~\eqref{eqs:7chainC} are length three chain complexes for any length six chain complex $(W,w_b,w_t,p)_{7}$ in~\cref{eqs:7chainA}.
\end{lem}
\begin{proof}
  In addition to~\cref{eq:CSS3chain}, that is again satisfied like in the previous lemma, one needs to show that
  \begin{equation}%\label{eq:}
    (K_{p-2}^T\boxplus K_{p+3})^T(K_{p-1}\oplus K_{p+2}^T)=0
  \end{equation}
  holds. It immediately follows from~\eqref{eqs:7chainA} since
  \begin{align}%\label{eqs:}
    K_{p-2}K_{p-1} & =0, \\
    K_{p+2}K_{p+3} & =0.
  \end{align}
  \cref{eqs:7chainC} is obtained by splicing which preserves the commutation relations.
\end{proof}
For a few examples see~\cref{table:folding}. It is not immediately obvious how to generalize the presented method to have two metachecks (chain complexes of length five) while having a non-trivial CSS code.

\subsubsection{Code examples obtained by splicing or folding}\label{sec:actI}

In this section I show several examples of non-trivial codes obtained from the methods developed so far. The rationale behind the choice of Coxeter groups is  to represent a diverse behavior of the resulting codes but it barely scratches the surface. Some of their small-scale examples were exemplified in~\cref{sec:CSScodes}. The Coxeter groups are the $A$ class (the symmetric groups), direct products of $C_2$ as a reducible example, a triangle hyperbolic group $\Delta_{2,3,7}$, a more exotic hyperbolic group with four generators, $(\boxtimes,\{s_1,s_2,s_3,s_4\})$~\cite{humphreys1990reflection}, whose Coxeter matrix is populated by 3's (except for the diagonal) so its Coxeter graph is a complete graph on four vertices -- $\boxtimes$, and the biggest exceptional group $E_8$.

\begin{center}
    \renewcommand{\arraystretch}{1.2}
    \extrarowheight=\aboverulesep
    %\addtolength{\extrarowheight}{\belowrulesep}
    \aboverulesep=0pt
    \belowrulesep=0pt
    \begin{table}[h]
           \begin{tabular}{@{}>{\columncolor{white}[0pt][\tabcolsep]}  *{4}c @{}}
           \toprule
            {\cellcolor{lightgray}\ $(W,w_b,w_t,p)_{3}$ } & $[n,k,d]$  &    $\max_{v}\limits{[\wt{h^v_{X/Z}}]}$      & \ \ method\ \ \\
                        \midrule
             {\cellcolor{lightgray}\ $(A_4,\id,\hat{t}_{A_4},5)$}       &   $[20,5,3]$                    &   9   &  crown    \\
             {\cellcolor{lightgray}\ $(A_5,\id,\hat{t}_{A_5},7)$}       &   $[101,25,4]$                  &  28    &   crown    \\
             {\cellcolor{lightgray}\ $(A_6,\id,\hat{t}_{A_6},10)$}      &   $[571,199,5]$                 &  76  &  crown  \\
             {\cellcolor{lightgray}\ $(C_2^{\times8},\id,\prod_{i=1}^{8}s_i$,4)}           &  $[69,5,5]$    &  9   &  crown  \\
             {\cellcolor{lightgray}\ $(C_2^{\times8},\id,\prod_{i=1}^{8}s_i$,4)}           &  $[70,8,5]$    &  8    &  $S^2$  \\
             {\cellcolor{lightgray}\ $(C_2^{\times10},\id,\prod_{i=1}^{10}s_i,5)$}      &  $[252,36,6]$     &  32   &  crown  \\
             {\cellcolor{lightgray}\ $(C_2^{\times10},\id,\prod_{i=1}^{12}s_i,5)$}       &  $[245,12,5]$   &  6   &  diamond  \\             
             {\cellcolor{lightgray}\ $(C_2^{\times12},\id,\prod_{i=1}^{12}s_i,6)$}       &  $[924,185,7]$   &  61   &  crown  \\
             {\cellcolor{lightgray}\ $(C_2^{\times12},\id,\prod_{i=1}^{12}s_i,6)$}       &  $[917,26,6]$   &  7   &  diamond  \\
             {\cellcolor{lightgray}\ $(\Delta_{2,3,7},\id,(s_1s_2s_3)^{10},24)$}      &  $[327,78,5]$       &  49   &  crown  \\
             {\cellcolor{lightgray}\ $(\boxtimes,\id,(s_1s_2s_3s_4)^3, 7 )$}      &  $[279,48,4]$           &  60   &  crown  \\
             {\cellcolor{lightgray}\ $(E_8,\id,(\prod_{i=1}^{8}s_i)s_1, 4 )$}      &  $[90,14,4]$           &  12   &  crown  \\
             {\cellcolor{lightgray}\ $(E_8,\id,(\prod_{i=1}^{8}s_i)s_1, 4 )$}      &  $[91,7,5]$            &  12   &  crown  \\
             \bottomrule
             \hline
            \end{tabular}\\ \vskip .3cm
    \caption{CSS codes obtained by three different methods: $S^1$ (crown) or $S^2$ splicing and diamond ($S^0$) removal of the checks of a trivial code  $(W,w_b,w_t,p)_{3}$ for a Coxeter system $W$ and its various 3-layer subposet intervals of $(w_b,w_t)$. The initial trivial codes are in the first column ($\hat{t}_{A_i}$ denotes the longest element of $A_i$) and the second column contains the derived codes. The third column is the maximal encountered stabilizer weight. The splicing weights seem outrageous but they are highly irregular -- see the main text,~\cref{table:crudeSplicing} and~\cref{fig:histogram}. All code distances were verified by the exact distance claculation tools mentioned in the introduction.}
    \label{table:S1S2splicing}
    \end{table}
\end{center}
The result of $S^1$ (crown) and $S^2$ splicing for  $(W,w_b,w_t,p)_{3}$ CSS codes derived from these groups is in~\cref{table:S1S2splicing} for the poset layers $\{l_{p-1},l_p,l_{p+1}\}$ of various $p$'s. There are several interesting observations. High-rate codes with decent distances were obtained. In fact, in both cases of Coxeter group families ($\{A_i\}$ and $\{C_2^{\times2i}\}$) everything seems to suggest that with an increasing group size the corresponding CSS codes' rate remains constant (or increases) and the distance increases as well. The price to pay is quite heavy stabilizers but a closer analysis uncovers that there is always a small number of very heavy stabilizers and a large number of very light ones. This is a promising starting point for weight reduction and I take first steps in this direction by developing a weight reduction procedure in the next section.

\begin{center}
    \renewcommand{\arraystretch}{1.2}
    \extrarowheight=\aboverulesep
    %\addtolength{\extrarowheight}{\belowrulesep}
    \aboverulesep=0pt
    \belowrulesep=0pt
    \begin{table}[h]
           \begin{tabular}{@{}>{\columncolor{white}[0pt][\tabcolsep]}  *{4}c @{}}
           \toprule
            {\cellcolor{lightgray}\ $(W,w_b,w_t,p)_{3}$ } & $[n,k,\{d_X,d_Z\}]$       &    $\max_{v}\limits{[\wt{h^v_{X/Z}}]}$    & splicing \\
                        \midrule
             {\cellcolor{lightgray}\ $(A_4,\id,\hat{t}_{A_4},5)$}       &   $[22,2,\{3,3\}]$             &  10  &  rnd    \\
             {\cellcolor{lightgray}\ $(A_6,\id,\hat{t}_{A_6},9)$}       &   $[531,22,\{\setlength{\fboxsep}{2pt}\colorbox{magenta!20}{$8\leq d_X\leq67$},5\}]$          &  21  &   rnd    \\
             {\cellcolor{lightgray}\ $(A_6,\id,\hat{t}_{A_6},10)$}      &   $[573,20,\{5,5\}]$           &  20  & rnd  \\
             {\cellcolor{lightgray}\ $(A_7,\id,\hat{t}_{A_7},14)$}           &  $[3736,94,\{5\leq d_X\leq6,5\}]$                      & 28 &  rnd  \\
             {\cellcolor{lightgray}\ $(C_2^{\times8},\id,\prod_{i=1}^{8}s_i$,4)}           &  $[70,16,\{5,5\}]$          & 10 &  rnd  \\
             {\cellcolor{lightgray}\ $(C_2^{\times10},\id,\prod_{i=1}^{10}s_i,5)$}      &  $[252,42,\{6,6\}]$            & 12 &  rnd \\
             {\cellcolor{lightgray}\ $(C_2^{\times12},\id,\prod_{i=1}^{12}s_i,6)$}      &  $[924,134,\{7,7\}]$           & 14 &  rnd  \\
             {\cellcolor{lightgray}\ $(C_2^{\times14},\id,\prod_{i=1}^{14}s_i,7)$}      &  $[3432,428,\{\leq8,\leq8\}]$          & 16 &  rnd  \\
             {\cellcolor{lightgray}\ $(C_2^{\times16},\id,\prod_{i=1}^{16}s_i,6)$}      &  $[8008,821,\{\leq7,\leq11\}]$         & 22 & rnd  \\
             {\cellcolor{lightgray}\ $(C_2^{\times16},\id,\prod_{i=1}^{16}s_i,7)$}      &  $[11440,1003,\{\leq8,\leq10\}]$       & 20 &  rnd  \\
             {\cellcolor{lightgray}\ $(C_2^{\times16},\id,\prod_{i=1}^{16}s_i,8)$}      &  $[12870,1432,\{\leq9,\leq9\}]$        & 18 &  rnd  \\
             {\cellcolor{lightgray}\ $(C_2^{\times18},\id,\prod_{i=1}^{18}s_i,9)$}      &  $[48620,4862,\{\leq10,\leq10\}]$      & 20 &  rnd  \\
             {\cellcolor{lightgray}\ $(\Delta_{2,3,7},\id,(s_1s_2s_3)^{10},24)$}      &  $[328,18,\{4,4\}]$              & 19 & rnd  \\
             {\cellcolor{lightgray}\ $(\Delta_{2,3,7},\id,(s_1s_2s_3)^{10},24)$}      &  $[328,172,\{4,4\}]$             & 36 & rnd twice  \\
             {\cellcolor{lightgray}\ $(\boxtimes,\id,(s_1s_2s_3s_4)^3, 7)$}      &  $[280,39,\{4,4\}]$                   & 24 & rnd   \\
             {\cellcolor{lightgray}\ $(\boxtimes,(s_1s_2s_3s_4)^3,(s_1s_2s_3s_4)^5, 5)$}      &  $[521,68,\{4,4\}]$      & 26 & rnd $Z$ stabilizer \\
             \bottomrule
             \hline
            \end{tabular}\\ \vskip .3cm
    \caption{Random splicing of some of the Coxeter groups in~\cref{table:S1S2splicing} allows the investigation of larger code instances. Despite the results being, in general, worse both in terms of the rate and distance (with a few exceptions), as a proxy to crown splicing it points to an improvement as the size of Coxeter groups increases. The maximal stabilizer weights are smaller but they are much more uniform and in some cases nearly constant, see~\cref{fig:histogram} for a generic situation. I verified the distance of some codes by the discussed exact methods, the rest remains upper-bounded. The upper bound of the significant asymmetry in $A_6$ highlighted in {\color{magenta!80}{magenta}} doesn't go away no matter how hard I sample. The exact methods are responsible for the lower bound.}
    \label{table:crudeSplicing}
    \end{table}
\end{center}

\cref{table:crudeSplicing} shows the results of random splicing. The main difference compared to~\cref{table:S1S2splicing} is the lower maximal stabilizer weight and that might seem as advantageous. However, the weight distribution is more (sometimes completely) uniform and on average significantly higher, see an example for the smallest code $A_4,p=5$ in~\cref{fig:histogram}, which is quite generic. A second noteworthy difference is that unlike crown splicing, the rate within a Coxeter group family $\{A_i\}$ or $\{C_2^{\x2i}\}$ is decreasing. All this points to an exciting but so far inconclusive eventuality that, indeed, the $S^1$ (and perhaps $S^2$) sphere internal structure does play a role in the QEC capabilities of the spliced codes  $(W,w_b,w_t,p)_{3}$ and the fast but random splicing is a mere proxy.

\begin{figure}[t]
  \resizebox{9cm}{!}{\includegraphics{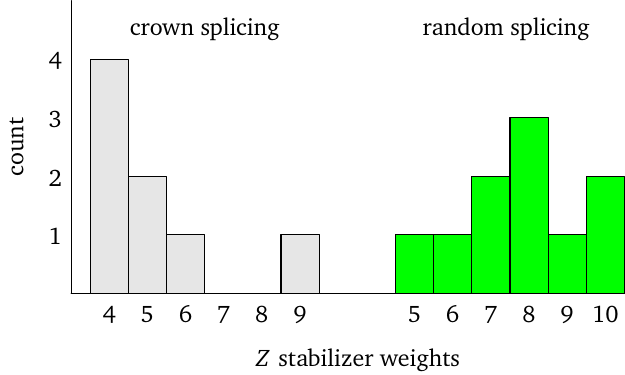}}
  \caption{The $Z$ stabilizer weight distribution for crown and random splicing of  $(A_4,\id,\hat{t}_{A_4},5)$, see the first row of~\cref{table:S1S2splicing} and~\ref{table:crudeSplicing}. Crown splicing produces highly non-uniform stabilizers that are more amenable to weight-reduction as shown in the last example of~\cref{sec:weightLoss}. }
  \label{fig:histogram}
\end{figure}

\begin{center}
    \renewcommand{\arraystretch}{1.2}
    \extrarowheight=\aboverulesep
    %\addtolength{\extrarowheight}{\belowrulesep}
    \aboverulesep=0pt
    \belowrulesep=0pt
    \begin{table}[h]
    \resizebox{\textwidth}{!}{
           \begin{tabular}{@{}>{\columncolor{white}[0pt][\tabcolsep]}  *{5}c @{}}
           \toprule
            {\cellcolor{lightgray}\ $(W,w_b,w_t,p)_{m}$ } & $[n,k,\{d_X, d_Z\}]$ & $\max_{v}\limits{[\wt{h^v_{X/Z}}]}$   & metacheck code\\
                        \midrule
             {\cellcolor{lightgray}\ $(C_2^{\times8},\id,\prod_{i=1}^{8}s_i,4)_5$}        &  $[112, 34,   \{6, 4\}]$    &   12     &  $\boldsymbol{\times}$  \\
             {\cellcolor{lightgray}\ $(C_2^{\times12},\id,\prod_{i=1}^{12}s_i,6)_5$}      &  $[1584, 417, \{6, 8\}]$  &    16     &   $\boldsymbol{\times}$ \\
             {\cellcolor{lightgray}\ $(C_2^{\times12},\id,\prod_{i=1}^{12}s_i,6)_7$}   &  $[1584,252, \{6,12\}]$   &  12  &  $[990,111,\{7\leq d\leq8,9\}]$  \\
             {\cellcolor{lightgray}\ $(C_2^{\times14},\id,\prod_{i=1}^{14}s_i,7)_5$}   &  $[6006,1497,\{\leq9,\leq7\}]$  &    18     &   $\boldsymbol{\times}$ \\
             {\cellcolor{lightgray}\ $(C_2^{\times14},\id,\prod_{i=1}^{14}s_i,7)_7$}   &  $[6006,924,\{\leq14,\leq7\}]$   & 14 & $[4004,428,\{\ll\setlength{\fboxsep}{2pt}\colorbox{orange}{388},\leq10\}]$  \\
             {\cellcolor{lightgray}\ $(C_2^{\times16},\id,\prod_{i=1}^{16}s_i,8)_5$}   &  $[22880,5434,\{\leq8,\leq10\}]$  &    20   &  $\boldsymbol{\times}$  \\
             {\cellcolor{lightgray}\ $(C_2^{\times16},\id,\prod_{i=1}^{16}s_i,8)_7$}   &  $[22880, 3432,\{\leq8,\leq16\}]$   & 16 &  $[16016,1639,\{\ll\setlength{\fboxsep}{2pt}\colorbox{orange}{1991},\leq11\}]$  \\
             {\cellcolor{lightgray}\ $(C_2^{\times18},\id,\prod_{i=1}^{18}s_i,9)_5$}   &  $[87516,12870,\{\leq18,\leq9\}]$  &    18   &  $\boldsymbol{-}$  \\
             {\cellcolor{lightgray}\ $(A_7,\id,\hat{t}_{A_7},14)_5$}                      &  $[7472, 206,  \{\leq6, \leq5\}]$  &   26     &  $\boldsymbol{\times}$  \\
             \bottomrule
             \hline
            \end{tabular}}\\ \vskip .3cm
    \caption{CSS codes ($m=5$) and CSS codes with a metacheck ($m=7$) from  chain complex folding. The highly asymmetric distances highlighted in {\color{orange}{orange}} are  spurious and of low/no confidence due to the codes' sheer size. They are provided by~\cite{QDistRnd} but for an insufficient sample size. Many of the codes in this table were not possible to  verify exactly.}
    \label{table:folding}
    \end{table}
\end{center}

The third line from the bottom  shows double random splicing as a test how it can be pushed, resulting in the highest rate but also quite heavy stabilizers. Sometimes just one-sided splicing already generates interesting codes like I show in the last row.

In~\cref{table:folding} there are CSS codes ($m=5$) and CSS codes with a single metacheck ($m=7$) obtained as a result of the chain complex folding. The distance of some metacheck codes  is highly asymmetric and the numbers in orange are likely very loose upper bounds. This is thanks to the size of the codes and it surely is just an artifact of poor sampling using~\cite{QDistRnd}. Exact distance estimation methods used here do confirm certain asymmetry in the distance (like for $C_2^{\x12},m=7$). This is potentially interesting but, unfortunately, the higher instances are computationally demanding and it remains to be shown if the asymmetry persists as suggested by the probabilistic distance estimates.

One could speculate that there are again strong hints of a CSS code family for the Coxeter group series $\{C_2^{\times 2i}\}$, whose rate is constant and the distance increases with $i$ while the stabilizer weight remains uniform and relatively low. Note that weight difference between the $X$ and $Z$ stabilizers is often high in folded codes and this `gradient' can again be used to efficiently decrease the greater of the two. To see how efficient this really can be using the developed weight-reduction method will be investigated elsewhere.

\section{Stabilizer weight reduction}\label{sec:weightLoss}

Splicing, in general, increases the weight of the CSS stabilizers. Many weight-reduction procedures have been proposed~\cite{hastings2021quantum,hsieh2025simplified,tan2025effective,sabo2024weight} and it remains an active research area~\cite{yuan2026quantum}. I will add  to the list the first steps of another weight-reduction method motivated by the properties of the Bruhat-based CSS codes (both trivial and spliced). It can, nonetheless, be applied to any CSS code.

The  commutativity condition of any qubit CSS code is equivalent to the existence of a chain complex of length two with $\bbZ_2$-valued modules (or Abelian groups), the associated boundary maps and their duals. If a chain complex has its origin in a tesselated or cellulated two-dimensional closed manifold (non-orientability comes automatically with the ring $\bbZ_2$) then the boundary maps have a clear interpretation as incidence relations mapping 2-cells (faces) to 1-cells (edges) to 0-cells (vertices). Its dual picture exchanges the role of the faces and vertices. Nothing much changes if the manifold is higher-dimensional and one may still call the $(p+1)$-cells faces, $p$-cells edges and $(p-1)$-cells vertices even though they don't resemble them and their incidence relations don't look familiar either.

For general CSS codes, however, no such geometric interpretation is necessary or even possible but the chain complex interpretation is always available. So I can call and mainly treat the modules' basis elements as (generalized) vertices, edges and dual vertices (`faces') because a check (both $X$ and $Z$) \emph{locally} looks like a vertex whose number of incident edges (physical qubits) equals the vertex weight. By the same token, even subsets of the incident edges are, again locally, incident to one or more faces (dual vertices). A chain complex as a general CSS code  is not geometrically embeddable as a nice 2-skeleton (see~\cref{app:CWcomplexes}) due to the fact that, for example,  a generalized edge can be shared by more than two vertices, or, dually, a face does not have to be demarcated by edges. Indeed, the hypergraph picture comes to mind (although it is probably different from the treatment color codes received in~\cite{delfosse2014decoding}) but here is no need to formally introduce these structures.

Where does it leave  $(W,w_b,w_t,p)_{3}$ CSS codes  and their spliced versions as the main protagonist of this work?  On the one hand, as shown in~\cref{def:BruhatChains}, the bases of the modules $C_p$ are $p$-cells, where~$p$ can be arbitrarily high. So there is an obvious geometric embedding -- it is just not something one can visualize. On the other hand, there exists a $S^2$ decomposition of every 3-level Bruhat poset in terms of the 0-, 1- and 2-cells, see~\cref{app:sphereClassification} or an example in~\cref{fig:A3S2decompositionCW}. The 1-cells are ordinary edges for each $S^2$ (they are incident to two vertices since the 1-skeleta of the $S^2$ CW spheres are ordinary graphs with multiple edges) but they routinely appear in more than one $S^2$ and often are incident to a different pair of vertices. The spliced $(W,w_b,w_t,p)_{3}$ CSS codes surely do not inherit  the elegant sphere decomposition anyway and  I will simply treat them as general CSS codes.

\begin{defi}\label{def:bridgedStarGraph}
  Let $\bigstar_m$ denote the $m$-star graph (a vertex~$v$ of degree $\d(v)=m$). By the \emph{bridged star graph} $\bigstar_{m_1,m_2}$ associated to~$v$ I shall call a graph defined as two stars, $\bigstar_{m_1}$ and $\bigstar_{m_2}$, whose central vertices share an edge (the bridge). So $m=m_1+m_2-2$ and for a given $m$ there exists several choices of $m_1,m_2$ which meaningfully start from $m_1,m_2\geq 3$.

  More generally, the \emph{$(m_1,m_2,\dots m_{b+1})$-bridged star graph} $\bigstar_{m_1,m_2,\dots,m_{b+1}}$ associated to a vertex $v$ of $\d(v)=m$ is a collection of $b+1$ $m_i$-star graphs, where the central vertices of pairs of  $m_i$- and $m_{i+1}$-star graphs for $1\leq i\leq b$ share an edge (the $i$-th bridge). It holds that
  \begin{equation}\label{eq:bridgeGraphConfigurations}
    m=\sum_{i=1}^{b+1}m_i-2b
  \end{equation}
  and there is a large variety $m_i\geq3$ satisfying the equation. For an illustration see~\cref{fig:bridge}.
\end{defi}
\begin{figure}[t]
  \resizebox{12.3cm}{!}{\includegraphics{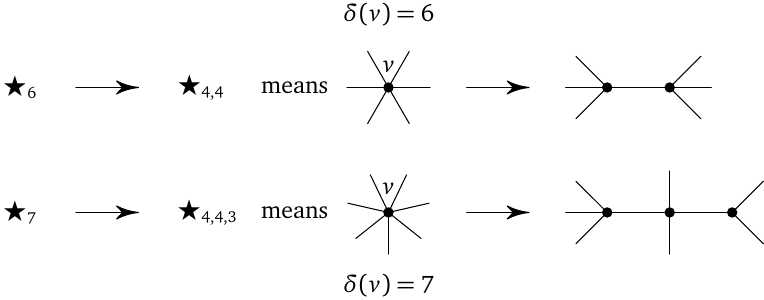}}
  \caption{Vertex reduction by vertex substitution using the bridged star graphs introduced in~\cref{def:bridgedStarGraph}. (Top) Vertex split of a degree six vertex using one bridge. (Bottom) Vertex split of a degree seven vertex using two bridges.}
  \label{fig:bridge}
\end{figure}

Every $X$ or $Z$ check of weight $m$ can be associated with $\bigstar_m$ for any CSS code by interpreting the data qubits of the PCMs as edges and the checks as vertices.

\begin{prop}\label{prop:weightLoss}
  Let $[n,k,d]$ be a CSS code and let $h^v_{X/Z}$ denote an $X/Z$ PCM row. Then, every stabilizer of $\wt{h^v}=m\geq5$ can be weight-reduced to $h^{v_1},h^{v_2}$ such that $\wt{h^{v_i}}\geq 3$ and $\wt{h^{v_j}}\geq 4,i\neq j$, giving rise to a new CSS code $[n+1,k,d']$. Moreover, if the number of edges incident to $v$ shared with every dual vertex is equal to at most two then the vertex~$v$ can be weight-reduced to  $\wt{h^{v_i}}\geq 3$ for $i=1,2$.
\end{prop}
Note that the number of incident edges of the vertex~$v$ shared by any dual vertex $v^*_i$ must be even or an empty set. Otherwise it is not a CSS code to start with. I will not prove the distance properties but only conjecture that $d'\simeq d$ as long as the reduced weight is not pushed too low. This is suggested by extensive numerical testing but also by some recent results investigating precisely the trade-off of the stabilizer weight and the code distance~\cite{wang2026check}. The only requirement I impose now is $d'>2$, see one of the remarks following the proof.
\begin{proof}
  The task of weight reduction of a chosen heavy stabilizer $h^v$ consists of updating the PCMs $H_X,H_Z$  of the initial $[n,k,d]$ CSS code to $H'_X,H'_Z$ and showing that they define another CSS code.   As argued above, every check $h^v$ such that $\wt{h^v}=m$ is always a star graph $\bigstar_m$ of some complicated hypergraph. Let $E=\{e_i\},1\leq i\leq m$ be the set of all  edges incident to~$v$ and let $v^*_j$ be all dual vertices sharing $E_j\subseteq E$ such that $|E_j|$ is even. The method to weight-reduce $v$ starts by substituting the star graph $\bigstar_m$ representing~$v$ by a  $\bigstar_{m_1,m_2}$ graph from~\cref{def:bridgedStarGraph} for suitable $m_i$'s such that $m=m_1+m_2-2$. The key role is played by the new data qubit $\hat e$ -- the bridge connecting the smaller stars. $\hat e$ is therefore incident to $v_1$ and $v_2$. But it must also be incident to at least one dual vertex otherwise I get $d'=1$. Let $v^*_1$ be such dual vertex. By construction it is incident to
  \begin{equation}\label{eq:updatedE1}
    \hat E_1=E_1\cup\hat e.
  \end{equation}
  Then, in order to get valid PCMs $H'_X,H'_Z$ I must make sure that both $v_1$ and $v_2$ overlap on an even number of edges from the set~$\hat E_1$. Since they are already incident to $\hat{e}\in\hat{E}_1$ it is sufficient to split $E_1$ into two subsets $E^{o_1}_1,E^{o_2}_1$ both containing an odd number of edges. This is always possible since $|E_j|$ is even for all $j$.

  Recall that there may be more dual vertices $v^*_j$ and their $E_j$ sets for $j>1$. If one of the following options occurs
  \begin{align}\label{eqs:noaction}
    E_j\cap E_1&=\emptyset\\
    E_j & \subseteq E^{o_1}_1\\
    E_j & \subseteq E^{o_2}_1
  \end{align}
  then no action is needed. If, on the contrary,
  \begin{align}\label{eq:yesaction}
    E_j\cap E^{o_1}_1\neq\emptyset &\mbox{ and } E_j\cap E^{o_2}_1 \neq\emptyset
  \end{align}
  then the bridge $\hat{e}$ must become incident to $v^*_j$ as well. Finally, if only one condition of~\cref{eq:yesaction} is satisfied no action is again needed.

  Going back to $v^*_1$, the procedure fails when $|E_1|\geq4$ and $m=4$. For $m=4$ the only reasonable bridged star graph is $\bigstar_{3,3}$. The set $E_1$ is promoted to $\hat{E}$ like in~\eqref{eq:updatedE1}. However, I cannot split $E_1$ into two odd $E^{o_1}_1,E^{o_2}_1$, where $|E^{o_1}_1|=1,|E^{o_2}_1|=3$ or vice versa. This is because $\d(v_1)=\d(v_2)=3$ but $v_1,v_2$ already share $\hat{e}$. So there is literally no `space' for the bigger edge set $E^{o_2}_1$. The only possibility is $|E^{o_1}_1|=|E^{o_2}_1|=2$ but that leads to a violation of the commutativity condition as I showed earlier.

  If, however,  $|E_1|=2$ one can  proceed as before and as long as $|E_j|=2,\forall j$ the update rules follow the general prescription described above. The heavy vertex star graph $\bigstar_4$ is then substituted by $\bigstar_{3,3}$ and therefore $\wt{h^{v_i}}=3$ for $i=1,2$. For the general case $m\geq 4$ I thus obtain $\wt{h^{v_i}}\geq 3$ as claimed.

  Finally, notice that every new $\bigstar_{m_1,m_2}$ adds one more data qubit (the bridge) and so $n$ changes to $n+1$. The net rank of the new stabilizers $H'_X,H'_Z$ increases by one. Therefore the new code parameters are $[n+1,k,d']$. Similarly, the substitution of a sufficiently heavy vertex $v$ by $\bigstar_{m_1,m_2,\dots,m_{b+1}}$ changes the code parameters to $[n+b,k,d']$.
\end{proof}
I offer a few remarks before showing some examples.
\begin{rem}
  \cref{prop:weightLoss} shows how to (multiplicatively) decrease the weight of a single stabilizer. But  every new bridge qubit is also incident to at least one dual vertex and so the weight of the complementary stabilizers increases  (just additively -- by one). If the goal is to lower the weight of all stabilizers below a desired level it is necessary to iteratively run the weight-reduction procedure and I leave  for future research how to do it optimally.
\end{rem}
\begin{rem}
  The condition $|E_j|=2,\forall j$ holds for all for all $(W,w_b,w_t,p)_{3}$ CSS codes due to~\cref{prop:sphereClassification}~(a). This is not very interesting since there is no obvious need to weight-reduce a code with zero logical qubits. However, there are many non-trivial CSS codes (of a different origin) satisfying the condition, see one of the below examples. Even in this case it is not desirable to pursue the lowest possible reduced weight to be equal to three due to the requirement $d'>2$. According to~\cite{wang2026check} if \emph{all} stabilizer weights equal three there is no CSS code encoding logical qubits and having $d>2$ at the same time (it is, nevertheless, possible, if only some stabilizers are weight three).
\end{rem}
The first two examples are small non-Bruhat codes that can be analyzed by hand. The last example is probably the smallest interesting spliced code ($d\geq3$).
\begin{exa}\label{exa:smallCodesWeightRed}
    For an example of a CSS code for which $|E_j|=2,\forall j$ holds, consider the $[9,1,3]$ Shor code with the stabilizer generators
    \begin{subequations}\label{eqs:ShorGens}
    \begin{align}
        \{\textcolor{red}{x_i}\} &= \{X_1 X_2 X_3 X_4 X_5 X_6, X_4 X_5 X_6 X_7 X_8 X_9\},\\
        \{\textcolor{blue}{z_j}\} &= \{Z_1 Z_2, Z_2 Z_3, Z_4 Z_5, Z_5 Z_6, Z_7 Z_8, Z_8 Z_9\}.
    \end{align}
    \end{subequations}
    Let's weight-reduce both $X$ checks by the transformation $\bigstar_6\mapsto\bigstar_{4,4}$. The reduction of the first check $x_1$ is depicted in~\cref{fig:shorWR}, where the bridge qubit $\hat{e}=e_{10}$ is incident to two vertices, $v_1$ and $v_3$. I want it to be incident to two dual vertices. To this end, I collect all $E_j$ sets and promote two of them to their hat versions. The table below shows the choice of $\hat{E}_2$ and $\hat{E}_3$.
    \begin{center}
    \renewcommand{\arraystretch}{1.2}
    \extrarowheight=\aboverulesep
    %\addtolength{\extrarowheight}{\belowrulesep}
    \aboverulesep=0pt
    \belowrulesep=0pt
    \begin{table}[h]
           \begin{tabular}{@{}>{\columncolor{white}[0pt][\tabcolsep]}  *{3}c @{}}
           \toprule
             $E_i$ & \textcolor{blue}{$v_i^*$} & $\hat{E}_i$\\
                        \midrule
               $\{e_1,e_2\}$   & \textcolor{blue}{$v^*_1$}   &   \\
               $\{e_2,e_3\}$   & \textcolor{blue}{$v^*_2$}   & $\{e_2,e_{10},e_3\}$   \\
               $\{e_4,e_5\}$   & \textcolor{blue}{$v^*_3$}  & $\{e_4,e_{10},e_5\}$  \\
               $\{e_5,e_6\}$   & \textcolor{blue}{$v^*_4$}  & \\
             \bottomrule
            \end{tabular}\\ \vspace{.3cm}
    \caption{(First column) Edges $E_i$ incident to $v_1$  paired to show to which dual vertex~$v^*_i$ (color-coded $Z$ stabilizer $z_i$) they connect to (second column). $v_2^*,v_3^*$ become incident to the bridge qubit $\hat{e}=e_{10}$ (third column), see~\cref{fig:shorWR}.}
    \vspace{-.5cm}
    %\label{table:}
    \end{table}
    \end{center}
    \begin{figure}[t]
  \resizebox{8cm}{!}{\includegraphics{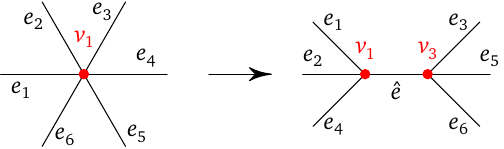}}
  \caption{Weight reduction of the stabilizer $x_1=X_1 X_2 X_3 X_4 X_5 X_6$ of Shor's code $[9,1,3]$ by $\bigstar_{6}\mapsto\bigstar_{4,4}$ of the corresponding vertex $v_1$.}
  \label{fig:shorWR}
\end{figure}
    This is a point where I could have made other choices. It wouldn't matter for the Shor code but generally it could lead to a CSS code with different parameters. Following~\cref{prop:weightLoss} I set
    \begin{align}%\label{eqs:}
      E^{o_1}_2=e_2, &\ E^{o_2}_2=e_3, \\
      E^{o_1}_3=e_4, &\ E^{o_2}_3=e_5
    \end{align}
    and make $v_1$ incident to $e_2$ and $e_4$ and $v_3$ incident to $e_3$ and $e_5$. Consequently, pairs of edges incident to $v_1$ or $v_3$ are also incident to $v_2^*$ or $v^*_3$. No further action is needed because~\cref{eq:yesaction} applies neither to $E_1$ nor to $E_4$. Hence, I get
    \begin{subequations}\label{eqs:ShorWeightReducedstep1}
    \begin{align}
    H'_X&=\left[
    \begin{array}{cccccccccc}
     1 & 1 & . & 1 & . & . & . & . & . & 1 \\
     . & . & . & 1 & 1 & 1 & 1 & 1 & 1 & . \\
     . & . & 1 & . & 1 & 1 & . & . & . & 1 \\
    \end{array}
    \right],\\
    H'_Z&=\left[
    \begin{array}{cccccccccc}
     1 & 1 & . & . & . & . & . & . & . & . \\
     . & 1 & 1 & . & . & . & . & . & . & 1 \\
     . & . & . & 1 & 1 & . & . & . & . & 1 \\
     . & . & . & . & 1 & 1 & . & . & . & . \\
     . & . & . & . & . & . & 1 & 1 & . & . \\
     . & . & . & . & . & . & . & 1 & 1 & .
    \end{array}
    \right].
    \end{align}
    \end{subequations}
    Repeating the same procedure for $x_2$ I finally obtain the following CSS code $[11,1,3]$:
    \begin{subequations}\label{eqs:ShorWeightReducedstep2}
    \begin{align}
    H''_X&=\left[
    \begin{array}{ccccccccccc}
     1 & 1 & . & 1 & . & . & . & . & . & 1 & .\\
     . & . & . & . & . & 1 & . & 1 & 1 & . & 1\\
     . & . & 1 & . & 1 & 1 & . & . & . & 1 & .\\
     . & . & . & 1 & 1 & . & 1 & . & . & . & 1\\
    \end{array}
    \right],\\
    H''_Z&=\left[
    \begin{array}{ccccccccccc}
     1 & 1 & . & . & . & . & . & . & . & . & .\\
     . & 1 & 1 & . & . & . & . & . & . & 1 & .\\
     . & . & . & 1 & 1 & . & . & . & . & 1 & .\\
     . & . & . & . & 1 & 1 & . & . & . & . & 1\\
     . & . & . & . & . & . & 1 & 1 & . & . & 1\\
     . & . & . & . & . & . & . & 1 & 1 & . & .
    \end{array}
    \right],
    \end{align}
    \end{subequations}
    where all $X$ checks are weight four as desired. Note that four $Z$ checks gained weight by one as anticipated but this is the minimal possible weight change for the transformation $\bigstar_6\mapsto\bigstar_{4,4}$ applied to two weight-six vertices.
\end{exa}
\begin{exa}
    By attempting  $\bigstar_4\mapsto\bigstar_{3,3}$ to transform one of the stabilizers $x=X_1X_2X_3X_4,z=Z_1Z_2Z_3Z_4$ of  the $[4,2,2]$ code, one finds that the only admissible split of $E_1=\{e_1,e_2,e_3,e_4\}$ is into two pieces of an even length and that fails to be a CSS code. On the other hand, its closest relative, the $[6,4,2]$ CSS code with $x=X_1X_2X_3X_4X_5X_6,z=Z_1Z_2Z_3Z_4Z_5Z_6$, can be weight-reduced by $\bigstar_6\mapsto\bigstar_{4,4}$ to a $[7,4,2]$ CSS code whose PCMs reads
    \begin{subequations}\label{eqs:642WeightReducedstep1}
    \begin{align}
    H'_X&=\left[
    \begin{array}{ccccccc}
     1 & 1 & 1 & 1 & 1 & 1 & 1
    \end{array}
    \right],\\
    H'_Z&=\left[
    \begin{array}{ccccccc}
     1 & 1 & 1 & . & . & . & 1\\
     . & . & . & 1 & 1 & 1 & 1
    \end{array}
    \right].
    \end{align}
    \end{subequations}
    Again, as anticipated, the $X$ check became heavier and so as an example of the iterative procedure mentioned in a previous remark I now apply the transformation $\bigstar_7\mapsto\bigstar_{4,5}$ resulting in a minimally modified CSS code
    \begin{subequations}\label{eqs:642WeightReducedstep2}
    \begin{align}
    H''_X&=\left[
    \begin{array}{cccccccc}
     1 & 1 & 1 & . & . & . & . & 1\\
     . & . & . & 1 & 1 & 1 & 1 & 1
    \end{array}
    \right],\\
    H''_Z&=\left[
    \begin{array}{cccccccc}
     1 & 1 & 1 & . & . & . & 1 & 1\\
     . & . & . & 1 & 1 & 1 & 1 & .
    \end{array}
    \right]
    \end{align}
    \end{subequations}
    with the parameters $[8,4,2]$ and $\wt{h^v}\leq 5$, where the increase by two physical data qubits is the best possible result. Even more weight-reduced modification of the $[6,4,2]$ code could be attempted by the repeated transformation $\bigstar_6\mapsto\bigstar_{3,3,3,3}$ of $x$ and $z$.

\end{exa}

\begin{exa}
  My last example is code $[20,5,\{3,3\}]$ generated by crown splicing of $(A_4,\id,\hat{t}_{A_4},5)_{3}$, where $\hat{t}_{A_4}=s_1s_2s_3s_4s_1s_2s_3s_1s_2s_1$ is the longest group element, already reported in~\cref{table:S1S2splicing} and compared in~\cref{fig:histogram} with random splicing. To see how the introduced weight reduction performs in realistic circumstances, I want to reduce the weight of $h_Z^3$, where $\wt h_Z^3=9$ as the single heaviest check (see the heavy blue stripe in~\cref{eqs:HZ2053}).
    \begin{subequations}%\label{eqs:}
    \begin{align}\label{eqs:HX2053}
      H_X &= \left[
        \begin{array}{cccccccccccccccccccc}
         1 & 1 & 1 & . & . & . & . & . & . & . & . & . & . & . & . & . & . & . & . & . \\
         . & . & 1 & 1 & 1 & 1 & . & . & . & . & . & . & . & . & . & . & . & . & . & . \\
         . & . & . & 1 & 1 & . & . & . & . & . & 1 & . & 1 & 1 & . & . & . & 1 & 1 & . \\
         1 & . & . & . & . & 1 & . & 1 & 1 & 1 & 1 & . & . & . & . & . & . & . & 1 & 1 \\
         . & . & . & 1 & . & 1 & . & . & 1 & . & . & 1 & 1 & 1 & 1 & 1 & . & . & . & . \\
         . & . & . & . & . & 1 & . & 1 & . & . & . & 1 & . & . & . & . & 1 & . & 1 & . \\
         . & 1 & . & . & . & . & 1 & . & . & . & . & . & 1 & . & 1 & . & . & . & . & 1 \\
        \end{array}
      \right],\\\label{eqs:HZ2053}
      H_Z & =  \left[
        \begin{array}{cccccccccccccccccccc}
         . & . & . & 1 & 1 & . & . & . & . & . & . & . & 1 & 1 & 1 & . & . & . & . & . \\
         1 & 1 & . & . & . & . & 1 & . & . & 1 & . & . & . & . & . & . & . & . & . & . \\
          \rowcolor{blue!50}
         . & 1 & 1 & 1 & . & . & . & . & . & . & 1 & 1 & 1 & 1 & . & . & . & 1 & 1 & . \\
         . & . & . & . & 1 & 1 & 1 & 1 & . & . & . & . & . & . & 1 & . & . & 1 & . & . \\
         . & . & . & . & . & . & . & . & . & 1 & . & 1 & . & 1 & . & . & . & . & 1 & . \\
         . & . & . & . & . & . & . & 1 & 1 & . & . & . & . & . & . & 1 & 1 & . & . & . \\
         . & . & . & . & . & . & . & . & . & . & 1 & 1 & 1 & . & . & . & 1 & . & . & 1 \\
         . & . & . & . & . & . & . & . & . & . & . & . & 1 & 1 & 1 & 1 & . & . & . & . \\
        \end{array}
      \right].
    \end{align}
    \end{subequations}
    I use $\bigstar_9\mapsto\bigstar_{5,6}$ and following~\cref{prop:weightLoss} I arrive at a new code $[21,5,\{3,3\}]$ in~Eqs.~\eqref{eqs:HXZprime2053}, none of its stabilizers is weight~9 or more. The split stabilizer is highlighted in light blue in~\cref{eqs:HZprime2053}, leading to $\wt h'^{3}_Z=5$ and $\wt h'^9_Z=6$, together with the new bridge qubit in grey:
    \begin{subequations}\label{eqs:HXZprime2053}
    \begin{align}\label{eqs:HXprime2053}
      H'_X &= \left[
        \begin{array}{cccccccccccccccccccc>{\columncolor{gray!20}}c}
     1 & 1 & 1 & . & . & . & . & . & . & . & . & . & . & . & . & . & . & . & . & . & 1 \\
     . & . & 1 & 1 & 1 & 1 & . & . & . & . & . & . & . & . & . & . & . & . & . & . & 1 \\
     . & . & . & 1 & 1 & . & . & . & . & . & 1 & . & 1 & 1 & . & . & . & 1 & 1 & . & 1 \\
     1 & . & . & . & . & 1 & . & 1 & 1 & 1 & 1 & . & . & . & . & . & . & . & 1 & 1 & . \\
     . & . & . & 1 & . & 1 & . & . & 1 & . & . & 1 & 1 & 1 & 1 & 1 & . & . & . & . & . \\
     . & . & . & . & . & 1 & . & 1 & . & . & . & 1 & . & . & . & . & 1 & . & 1 & . & 1 \\
     . & 1 & . & . & . & . & 1 & . & . & . & . & . & 1 & . & 1 & . & . & . & . & 1 & . \\
        \end{array}
      \right],\\\label{eqs:HZprime2053}
      H'_Z & =  \left[
        \begin{array}{cccccccccccccccccccc>{\columncolor{gray!20}}c}
      . & . & . & 1 & 1 & . & . & . & . & . & . & . & 1 & 1 & 1 & . & . & . & . & . & . \\
     1 & 1 & . & . & . & . & 1 & . & . & 1 & . & . & . & . & . & . & . & . & . & . & . \\
     \rowcolor{blue!20}
     . & . & 1 & . & . & . & . & . & . & . & 1 & . & . & . & . & . & . & 1 & 1 & . & 1 \\
     . & . & . & . & 1 & 1 & 1 & 1 & . & . & . & . & . & . & 1 & . & . & 1 & . & . & . \\
     . & . & . & . & . & . & . & . & . & 1 & . & 1 & . & 1 & . & . & . & . & 1 & . & . \\
     . & . & . & . & . & . & . & 1 & 1 & . & . & . & . & . & . & 1 & 1 & . & . & . & . \\
     . & . & . & . & . & . & . & . & . & . & 1 & 1 & 1 & . & . & . & 1 & . & . & 1 & . \\
     . & . & . & . & . & . & . & . & . & . & . & . & 1 & 1 & 1 & 1 & . & . & . & . & . \\
     \rowcolor{blue!20}
     . & 1 & . & 1 & . & . & . & . & . & . & . & 1 & 1 & 1 & . & . & . & . & . & . & 1 \\
        \end{array}
      \right].
    \end{align}
    \end{subequations}
     The bridge qubit must necessarily be  stabilized by some $X$ checks (here four) but this by itself did not increase the weight beyond the target of nine (otherwise there would be no gain). This see--saw effect is not welcomed but not really forbidden. When decreasing the weight of more than one stabilizers, it gets sometimes worse before it gets better. But here I managed to avoid it.
\end{exa}

\section{Conclusions and open problems}\label{sec:concl}

In this work I present several methods to turn Bruhat order of Coxeter groups into CSS codes. This is possible thanks to the dual role Bruhat order plays. What starts as a poset of Coxeter group elements ordered according to whether they are related a general reflection (a purely geometrical property) is known to be, in fact, a face poset ordering $p$-cells ($p$-faces) of a high-dimensional manifold by inclusion (a topological property)~\cite{bjorner1984posets}. Hence, Bruhat order can be seen as a large cellular chain complex (arbitrarily large for infinite Coxeter groups), namely a regular CW complex. It is this interpretation of Bruhat order that I turn into non-trivial CSS codes. However, this procedure is not entirely straightforward since the cellulated manifold is always a sphere of an arbitrary high dimension, which is always topologically trivial. As such, the corresponding CSS code encodes zero logical qubits. I use a unique property of Bruhat order, called the $S^k$ sphere decomposition for $k=0,1,2$, to probabilistically turn trivial codes into CSS codes with typically very good code rates and decent distance parameters. To this end, I introduce a  procedure I call splicing, which is a simple stabilizer transformation using the information provided by the $S^1$ or $S^2$ sphere decomposition. Despite its simplicity it provides non-trivial CSS code for any Coxeter group. One drawback is the existence of a few high-weight stabilizers. This can be dealt with by a weight-reduction method I introduce as well but postpone its iterative and asymptotic performance investigation for future research. One can envision a more sophisticated use of the sphere decomposition information and I~leave it as an open question as well.

My next strategy to create interesting CSS codes relies on a deterministic transformation of longer CW chain  complexes into shorter ones of length two and three through the procedure I call chain complex folding. Length three chain complexes are interpreted as CSS codes with a metacheck. Chain complex folding can be considered as a special case of splicing but this time the stabilizer weight of the created codes is better controlled.  I conjecture the existence of multitudes of CSS codes as well as their large finite and infinite families thanks to their Coxeter group origin.

Many of the large code instances have their distance merely upper-bounded. This is related to the main open question: does a sphere decomposition points to non-trivial codes and if yes is it possible to generate them deterministically including an analytical lower bound for the distance? Given the sphere interpretation as simple classical  and quantum repetition codes it is not unthinkable that this type decomposition may shed light on otherwise hard analytical CSS code properties.

Another open question is what makes a Coxeter group a good source of codes. I witnessed that some groups provide better codes than the others either via splicing or folding. In fact, the easiest group to work with is the reducible group family $\{C_2^{\x2i}\}$, whose Bruhat order coincides with the weak order (that is, a Cayley graph/poset). Cayley graphs figure in almost all novel QEC codes constructions in some capacity but the one I present here does not become one of them once the Bruhat and Cayley graph coincide. So even in this case the spliced or folded CSS codes differ if the other methods were applied to Coxeter groups.

Finally, I haven't taken the opportunity to investigate the performance of the dicscovered CSS codes equipped with a metacheck. This deserves to be explored together with finding ways how to generate 5-term chain complexes that correspond to $X$ and $Z$ metachecks.

\section*{Acknowledgments}

Several people shared their insightful comments at various stages of this work. They are (in alphabetical order): Benjamin Brown, Oscar Higgott, Angus Kan, Alex Neville, Brendan Pankovich, Guoming Wang and Joel Wallman, with special thanks to Leonid Pryadko for discussions about the distance estimation tools he and his collaborators developed.

\section{Math background}\label{sec:Apps}

\appendix

\section{Discrete groups}\label{app:discreteGroups}

Any discrete group can be characterized in terms of a presentation $\langle S | R \rangle$, where $S$ is a set of generators and $R$ is a set of relations, that is, products of group elements that are equal to the identity. The free group over $S$ consists of all \emph{words}, that is, all ordered products of elements of $S$, where it is referred to each element of the product as a \emph{letter}. I can \textit{reduce} words by using the relations to rewrite subexpressions in fewer letters. A word is in reduced form if it cannot be shortened using the relations. The length $\ell_S(g)$ of a group element $g$ is defined to be the number of letters in a reduced word for $g$.

It is possible to determine the length of a group element using the Cayley graph~\cite{grossman1964groups} as follows. The left Cayley graph $C_{G, S}$ of a group $G$ with a generating set $S$ is a directed graph with $G$ as its vertex set and a directed edge set
\begin{align}\label{eq:leftCayley}
    \{(g, sg) | g\in G, s \in S \}.
\end{align}
Then $\ell_S(g)$ is the number of edges between the identity and $g$ in the Cayley graph.
%..define the length of the group to be
%\begin{align}
%    \ell(G,S) = \max_{g \in G} \ell_S(g).
%\end{align}
The Cayley graph encodes one partial order on the group (see \cref{app:partialOrder}) called \emph{weak order}, where $h$ covers $g$ if $hg^{-1} \in S$. One can similarly introduce the right Cayley graph and the corresponding partial order.

\section{Coxeter groups}\label{app:Coxeter}

A Coxeter group~\cite{coxeter1973regular,davis2012geometry} is a group with a presentation of the form
\begin{align}\label{eq:coxPresentation}
    W = \langle s_i |  (s_i s_j)^{m_{ij}} = \id \rangle,
\end{align}
where $m_{ii} = 1$ for all $i$ and I will omit the $\id$. Coxeter group is a distinguished class of discrete groups forming several finite and many infinite families. The generating set $S$ consists of simple reflections -- a term alluding to the geometrical origin of Coxeter groups as generalized reflections (hence the involution $m_{ii} = 1$ which I will also omit). A pair $(W, S)$ is then called a Coxeter system, where $W$ is a Coxeter group and $S$ is a set of generators for $W$. For any $s_i, s_j$ with no relation, I use $m_{ij} = \infty$ which necessarily generates an infinite group.  Assuming that the generators are independent, the off-diagonal elements must be at least two by the uniqueness of inverses. Let $i, j$ be such that $m_{i, j} = 2$. Then multiplying
\begin{align}
    (s_i s_j)^2=\id
\end{align}
from the left by $s_i$ and the right by $s_j$ gives
\begin{align}
    s_i s_j = s_i^2 s_j s_i s_j^2 = s_j s_i,
\end{align}
and so $s_i$ and $s_j$ commute. Thus, it is common to encode information about a Coxeter group into a Coxeter graph whose vertex set is $S$, $(s_i, s_j)$ is an edge if $m_{ij} \geq 3$ (that is, only the non-commuting generators are connected) and  the edges are labeled if $m_{ij} > 3$.

A prominent example of an infinite Coxeter family is the $A_n$ class also known as the symmetric groups $S_{n+1}$ of order $(n+1)!$. Case $n=3$ has the presentation
\begin{align}\label{eq:A3presentation}
\langle s_1, s_2,s_3|(s_1 s_2)^3, (s_1s_3)^2,(s_2s_3)^3\rangle,
\end{align}
and the Coxeter matrix
\begin{align}\label{eq:CoxMatA3}
    \begin{bmatrix}
        1 & 3 & 2 \\
        3 & 1 & 3 \\
        2 & 3 & 1
    \end{bmatrix}.
\end{align}
In the symmetric group case the simple reflections are the set of all elementary transpositions: $s_i=(i,i+1),1\leq i\leq n$.

Coxeter groups possess all sorts of remarkable properties and some of them will become handy in this work~\cite{davis2012geometry,humphreys1990reflection}. I simply list a few necessary facts that are too basic to be omitted or become useful later. Associated with the Coxeter matrix is the Schl\"{a}fli matrix $S=[s_{ij}]$ with elements
\begin{align}
    s_{ij} = \begin{cases}
    -\cos \frac{\pi}{m_{ij}} & m_{ij} \leq \infty \\ -1 & \mbox{otherwise.}
    \end{cases}
\end{align}
A Coxeter group is finite if and only if $S$ is positive definite, affine if $S$ is positive semi-definite, and indefinite otherwise. Each connected component of a Coxeter graph corresponds to a normal subgroup, since any product of generators from that component commutes with any product of generators outside the component. Thus, any Coxeter group can be regarded as the direct product of its connected components and a Coxeter group irreducible if its Coxeter graph is connected. Finite and affine irreducible Coxeter groups have been classified~\cite{humphreys1990reflection,davis2012geometry}. The finite irreducible Coxeter groups  classified and belong to one of four one-parameter families or six exceptional groups. These groups are the symmetry groups of specific regular polytopes. Less is known about the structure of indefinite Coxetere groups also called hyperbolic Coxeter groups. An important class of finite irreducible Coxeter groups is called  Weyl (or crystalographic) groups. In the complete classification of finite irreducible Coxeter groups the only non-Weyl groups are $H_3,H_4$ and $I(m)$ for $m=5,m\geq7$.

Let $(W,S)$ be a Coxeter system. The reflections of $(W,S)$ are defined to be the set $T=\cup_nT_n$, where
\begin{align}\label{eq:reflectionT}
    T_n = \{wsw^{-1} : w \in W, \ell(w)=n, s \in S\}.
\end{align}
Going back to the Coxeter system $(A_n,S)$ the set of reflections consists of all transpositions $(i,j)$. A reflection $t\in T$ acts as $w\mapsto wt$ but unlike a simple reflection there is no difference between the left and right action: for the right action of $t$ there is a different  $t'\in T$ where $t'w=wt$.

Coxeter groups can be characterized combinatorially and~\cite{bjorner2006combinatorics} provides plenty of details. One of the most important properties of Coxeter groups is called the strong exchange property (SEP). Let $w=s_1\dots s_r$ be a reduced word (so $\ell(w)=r$) and let $t\in T$. Then $\ell(wt)<\ell(w)$ iff $wt=s_1\dots\slashed{s}_i\dots s_r$ for $i\in[r]$. The practical importance of the SEP is that it provides  a feasible way of finding all reflections $t$ satisfying $\ell(wt)=\ell(w)-1$ for a chosen $w$. This procedure helps determine a Bruhat order interval of any finite or infinite Coxeter group that I will introduce in~\cref{app:partialOrder}.

\section{Partially ordered sets and Bruhat order}\label{app:partialOrder}

This section is heavy on terminology but not without purpose. I frequently use it to describe CSS codes parameters derived from Bruhat order. A \emph{poset}~\cite{stanley1997enumerative,davey2002introduction} (partially ordered set) is a set $P$ with a relation $\leq$ satisfying:
\begin{description}
  \item[Reflectivity] $p\leq p$ for all $p\in P$.
  \item[Transitivity] $r\leq s$ and $s\leq q$ implies $r\leq q$ for all $r,s,q\in P$.
  \item[Antisymmetry] $r\leq s$ and $s\leq r$ implies $r=s$ for $r,s\in P$.
\end{description}
A poset is thus a pair $(P,\leq)$ but I will often mildly abuse the notation and call $P$ the poset. It is said that $q\in P$ \emph{covers} $r\in P$ (written $r\lessdot q$) if $r<q$ and there is no $s\in P$ such that $r<s<q$ ($x<y$ means $x\leq y,x\neq y$). The \emph{closed interval} $[r,q]$ in a poset $P$ are all $s\in P$ such that $r\leq s\leq q$. The \emph{open interval} $(r,q)$ is all $s\in P$ such that $r<s<q$. A finite poset (a poset whose set is finite) can be depicted as a \emph{Hasse diagram}, where the vertices are  poset elements and the edges are cover relations. By convention, the bigger elements of $P$ are drawn at the top but I will almost always `topple' the poset by putting them to the right. In fact, a larger class of posets called \emph{locally finite posets} can be drawn as a Hasse diagram as well. These are infinite posets where every  interval is finite.

Elements $p,q\in P$ are \emph{incomparable} if neither of $p<q,q<p$ holds and \emph{comparable} otherwise.  A subset of $P$ is called a \emph{chain} if all its elements are comparable. So for a chain $X\in\{x_0,\dots,x_n\}\subset P$ its elements satisfy $x_0<x_1<\dots<x_n$, where  the \emph{length} of the chain is defined to be~$n$. If in addition  $x_0\lessdot x_1 \lessdot \dots\lessdot x_n$ holds then $X$ is \emph{maximal} (it cannot be refined). A \emph{bounded poset} $P$ is graded if all maximal chains are of the same length. A poset is called bounded if there exists a unique least and a unique greatest element. I will denote them $\hat{b}$ or $w_b$ as `bottom' and $\hat{t}$ or $w_t$ as `top'. A \emph{rank function} $\rank:P\mapsto\bbN$ is a function satisfying $\rank{x_{i+1}}=\rank{x_i}+1$ whenever  $x_i\lessdot x_{i+1}$. A poset equipped with a rank function is called a \emph{graded poset} and  the bottom chain element $x_0=\hat{b}$ gets assigned, by convention, $\rank{x_0}=0$. The length of a closed interval $[x,y]$ of a graded poset is defined to be $|[x,y]|\df\rank(y)-\rank(x)$. Note that $|[x,y]|\equiv|(x,y)|$ for an open interval $(x,y)$.

A subset $l_p=\{x_i\}_i\subset P$ is called a \emph{level}~$l_p$ (or sometimes layer) of rank $p$ if it contains \emph{all} incomparable elements of rank~$p$, that is, all $x_i$ such that $\rank{x_i}=p$. Let $(P,\leq)$ be a locally finite poset and let $Q\subseteq P$. Then $(Q,\leq)$ is called an \emph{induced subposet} (subposet for short) if it inherits the structure of the poset~$(P,\leq)$, that is, $u\leq v$ in $(Q,\leq)$ iff  $u\leq v$ in $(P,\leq)$.

Let $\{l_{p_i}\}$ be a collection of $n$ consecutive layers of rank $p_i$ considered to be a subposet $Q=\cup_il_{p_i}$ of a graded poset~$P$. I will call it an $n$-tuple or $n$-layer subposet. As an induced subposet its Hasse diagram is a subgraph of the Hasse diagram of the poset~$P$.

\subsection{Bruhat order on Coxeter group}\label{appsub:bruhatorder}

Bruhat order makes any Coxeter group into a poset with deep links to other areas of mathematics. Let $(W,S)$ be a Coxeter system. A covering relation $u \lessdot w$ is defined for $u,w\in W$ if $w^{-1}u \in T$ defined in~\cref{eq:reflectionT} whenever $\ell(w) = \ell(u) + 1$, where $\ell(u)$ is the length of the Coxeter word $u$. This covering relation can be extended to a partial order, known as the Bruhat order, on the Coxeter system. Equivalently~\cite{bjorner2006combinatorics}, it also holds $u \lessdot w$ if a reduced word for $u$ can be obtained by removing one letter from any reduced word for $w$. Thus, the Bruhat order is equivalent to $u \leq w$ if any reduced word for $w$ contains a (not necessarily consecutive) subsequence that is a reduced word for $u$.

The Bruhat order on Coxeter groups is a graded poset whose rank function is the length function. Hence the length of a closed  Bruhat interval $[x,y]$ is $\ell([x,y])=\ell(y)-\ell(y)$ and as $\ell([x,y])\equiv\ell((x,y))$.

\section{Finite regular CW complexes}\label{app:CWcomplexes}

CW complexes~\cite{whitehead1949combinatorial,lundell2012topology,munkres2018elements,cooke2015homology} offer a flexible and elegant way of building a large number of `nice' topological spaces. Unlike more familiar simplicial complexes, where $d$-dimensional simplices are glued together in quite a rigid way, the CW complexes are built from $d$-balls called $d$-cells. Often they are called just cell complexes nowadays although one has to be careful about the properties of the attaching map used to specify the way the cells are glued together. CW stands for `Closure-finite and Weak topology' -- the terminology mainly relevant when building infinite-dimensional topological spaces. Since all the complexes studied here are finite I will  focus on introducing only an absolutely necessary amount of information about CW complexes.  Besides their conceptual importance,  CW complexes offer amazing flexibility compared to simplicial or polyhedral complexes by economically building topological spaces or cellulating manifolds.

Note that CW complexes as topological spaces are more general than manifolds (topological spaces that are locally Euclidean). Even though the spheres $S^d$ as initial objects are a manifold it is not obvious whether after they are converted to CSS codes they still remain some sort of manifold. And that is fine because they remain CW complexes and all that is needed is to calculate its homology.

Let $B^d\df\{(x_1,\dots,x_d)\in\bbR^d;\sum_ix_i^2\leq1\}$ be called the $d$-dimensional \emph{ball} (or~$d$-\emph{ball} or $d$-\emph{disc}) and its boundary the $(d-1)$-dimensional sphere $S^{d-1}\df\{(x_1,\dots,x_d)\in\bbR^d;\sum_ix_i^2=1\}$. An open $d$-\emph{cell} $e^d$ of dimension~$d$ is defined as a topological space homeomorphic to $\inter{B^d}=B^{d}\backslash S^{d-1}$  (the interior of $B^d$) and $S^{d-1}=\partial B^d$ its boundary. Let $X$ be a topological space  (in particular Hausdorff) whose \emph{CW decomposition} is a collection of  open $d$-cells $\{e^d_\a\}_{\a\in\iota_d}$, where $\a$ is drawn from finite indexing sets~$\iota_d$. Then
\begin{equation}\label{eq:X}
  X=\bigcup_{0\leq k<\infty}\bigcup_\a e^k_\a
\end{equation}
and
\begin{equation}\label{eq:skeleton}
X^d\df\big\{\bigcup_\a e^k_\a;\a\in\iota_k,0\leq k\leq d\big\}
\end{equation}
is called the $d$-\emph{skeleton} of $X$. A space $X$ together with the cells is called a \emph{CW complex} if
\begin{enumerate}
  \item[1.] For each $e_\a^d$ there exists a \emph{characteristic map} $\chi^d_\a:B^d\mapsto X$,
  \item[2.] $\chi^d_\a(\inter{B^d})=e^d_\a$ is a homeomorphism for all $\a$ (an open $d$-disc maps to an open $d$-cell),
  \item[3.] $\chi^d_\a(\partial B^d)\subseteq X^{d-1}$,
\end{enumerate}
where item 3. hides what is behind the `C' in CW: it dictates that the cell's closure can't touch infinitely many other cells. The third item also defines the action of the \emph{attaching map} acting on the boundary $S^{d-1}$ of $B^d$. The index~$\a$ indicates that the attaching map may differ for different cells (including of the same dimension). As I mentioned, in a finite case there is no need to explain the `W' as it is automatically satisfied and would take us to deeper waters. The subspaces $\chi^d_\a(B^d)$ are called closed $d$-cells and are to be denoted~$\overline{e}^d_\a$.

A finite CW complex defined this way is still not nice enough for the purposes of code building. It is further required that the attaching map is an embedding (that is, an injective and continuous map)~\cite{cooke2015homology}, making the CW complex \emph{regular}. As it will become apparent in~\cref{appsub:CWhomology} the requirement of regularity is to make the CW complexes simplicial-like while keeping them strictly more powerful than simplicial or polyhedral complexes.

\begin{figure}[t]
  \resizebox{14cm}{!}{\includegraphics{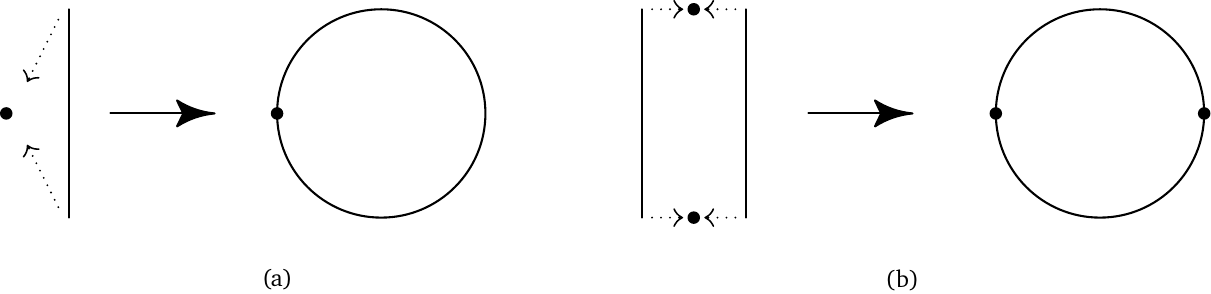}}
  \caption{How to make the $S^1$ sphere as a CW complex. (a) One 0-cell (the black dot) and one 1-cell (the line) are the building blocks. The dotted arrow shows the action of the attaching map~$\chi^1$ and the result is the CW complex $S^1=e^0\cup e^1$. (b) Two 0-cells and two 1-cells attached to them forming $S^1$ as a regular CW complex. In this case the attaching map is injective.}
  \label{fig:S1CW}
\end{figure}

Building spaces as CW complexes can be introduced in different ways but in the finite case the inductive procedure seems the most friendly. Loosely speaking, the creation of a CW complex is  about attaching cells of increasing dimensions: Let $X^0$ be a 0-skeleton, that is, a set of 0-cells (discrete points). The space $X^d$ is created by attaching $d$-cells $e^d_\a$ by virtue of the attaching map $\chi^d_\a: S^{d-1}\to X^{d-1}$ introduced earlier. One inductively forms  $(X^{d-1}\bigcup_\a B_\a^d)/(x\sim\chi_\a(x))$ as a disjoint union followed by taking the quotient by identifying the cell boundary point $x$ with its image under the attaching map $\chi_\a(x), \forall x$. The finite CW complex is then formed by setting $X=X^n$ for some $n<\infty$, which is called the CW complex dimension. For clarity, I show the first inductive step in more detail. Once the 0-skeleton $X^0$ is `deployed' I take a collection of 1-cells $\{e^1\}_\a$ together with the attaching maps $\chi^1_\a:S_{\a}^0\to X^0$. Recall that the attaching maps  prescribe what happens with all the 1-cells' boundaries. So the first inductive step consists of $X^0\bigcup_\a B^1_\a$ (`placing' the 0- and 1-cells next to each other) followed by identifying the endpoints with the 0-cells in $X^0$. In this way I  obtain a 1-skeleton~$X^1$ and the procedure starts over again\footnote{Note that the 1-skeleton of a regular CW complex is  a multiple-connected graph without self-loops.}. The described construction is a natural and versatile way of building complicated spaces from simple components ($d$-cells) that generalizes the simplicial and other complexes.

\begin{exa}
  As the standard example, let's build $S^1$ in two different ways. First, as an irregular CW complex, see Fig.~\ref{fig:S1CW}~(a). I take one 0-cell $e^0$ and one 1-cell $e^1$, where the attaching map merges the boundary of $e^1$  to $e^0$. I thus obtain~$S^1$ as a an irregular CW decomposition $S^1=e^0\cup e^1$. The irregularity comes from the fact that the $e^1$ endpoints map to the same 0-cell. For a regular cell decomposition let's have two 0-cells, $e^0_1$ and $e^0_2$. I now take two 1-cells, $e^1_1,e^1_2$, with the attaching maps illustrated in Fig.~\ref{fig:S1CW}~(b). I again obtain~$S^1$ but this time as a regular CW complex whose decomposition is $S^1=e_1^0\cup e_2^0\cup e_1^1\cup e_2^1\equiv2e^0\cup2e^1$. One can continue and attach two 2-cells to the $S^1$: $S^1\cup e_1^2\cup e_2^2=S^2$ to get a regular $S^2$ CW complex as shown in~\cref{fig:S2CW} (left).
\end{exa}

\begin{figure}[t]
  \resizebox{10.5cm}{!}{\includegraphics{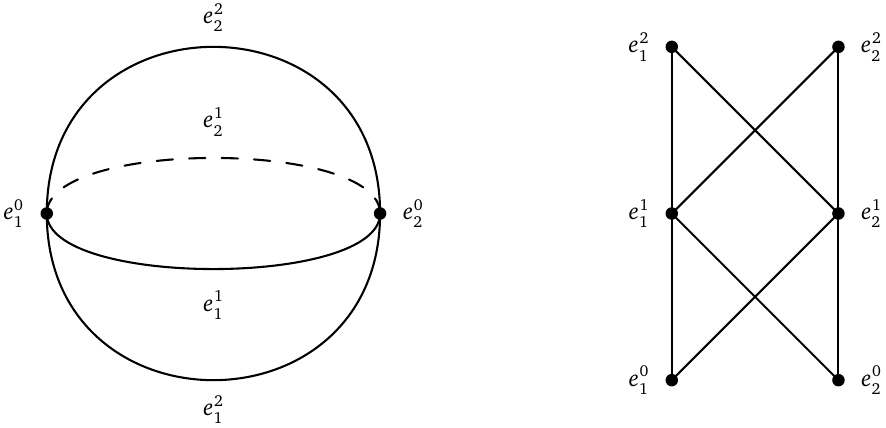}}
  \caption{(Left) Regular CW complex of the $S^2$ sphere as constructed in the main text. (Right) The corresponding face poset $\euF(S^2)$.}
  \label{fig:S2CW}
\end{figure}

\subsection{Homology of regular CW complexes}\label{appsub:CWhomology}

While on the conceptual level the homology of CW complexes is strictly more powerful than the one based on simplices or polygons, at the practical level (in the finite case and especially over $\bbZ_2$ that I follow in this section/work so that I can ignore any orientation information) there are  almost no differences when it comes to defining boundary operators acting on $p$-cells as the elementary building block of regular CW complexes. The homological machinery of regular CW complexes resembles the one  of simplicial complexes~\cite{cooke2015homology}.

Let $X$ be a regular CW complex or more precisely its underlying space. In anticipation of the things to come I terminologically look ahead and introduce the \emph{face poset} $\euF(X)=(X,\subseteq)$ as the set of closed cells $\overline{e}_\a$ ordered by inclusion, that is, $e_\a\leq e_\b$ whenever $\overline{e}_\a\subseteq\overline{e}_\b$ (see in~\cref{fig:S2CW} (right) for an example). Following the earlier poset terminology, I write $e_\a\lessdot e_\b$ if $e_\a<e_\b$ and when there is no $e$ such that $e_\a<e<e_\b$.

I define the incidence number $[e_\a:e_\b]$  to be equal to one iff $e_\b\lessdot e_\a$  and zero otherwise. Let
\begin{equation}\label{eq:dchains}
  c_p\df\sum_{\a\in\iota_p}\g_\a e^p_\a
\end{equation}
be a $p$-chain as a formal sum of $p$-cells where $\g\in\bbZ_2$. I write $C_p(X,\bbZ_2)$ as an Abelian group of all $p$-chains. It is possible to introduce a basis -- the $p$-cells themselves and promote $C_p(X,\bbZ_2)$ to a module (commutative ring-valued vector space). From there, one can define the boundary operator
\begin{equation}\label{eq:boundaryMap}
  \partial_p:C_p(X,\bbZ_2)\mapsto C_{p-1}(X,\bbZ_2)
\end{equation}
whose action is given by
\begin{equation}\label{eq:par}
  \partial_pe^p=\sum_{\b}[e^p:e_\b^{p-1}]e^{p-1}_\b=\sum_{e_\b\lessdot e}e^{p-1}_\b.
\end{equation}
We then have the orthogonality result that $\partial_{p-1} \partial_p = 0$:
\begin{align}\label{eq:parpar}
\partial_{p-1} \partial_p e^p & = \sum_{e_\b\lessdot e} \partial_{p-1} e^{p-1}_\b=0,
\end{align}
where the first line follows from linearity. For the second row, it can be shown~\cite{cooke2015homology} that for regular CW complexes there are always exactly two cells whose boundary is some $e^{p-2}$. Hence, just like for simplicial complexes, it will appear in exactly two of the $e^{p-1}_\b$ terms in \cref{eq:parpar} and thus will cancel out as addition is modulo 2.

The topology of space is its intrinsic property which is often hard to determine. A suitable homology theory has to satisfy certain axioms (formulated by Eilenberg and Steenrod~\cite{eilenberg1945axiomatic}) and CW homology is one such theory. An added bonus is that it is more suitable for practical calculations compared to the homology based on simplicial complexes. But whatever tool is chosen the axiom uniqueness property makes sure that one arrives at the same conclusion. Working with chains made of $p$-cells allows us to define useful concrete structures to be able to uncover the space's topological properties. One of them is a chain complex $C_\bullet=(C_p,\partial_p)$ which, at an abstract level, is a collection of Abelian groups $C_p$ ($0\leq p\leq d$) together with the homomorphisms $\partial_p:C_p\subset C_{p-1}$ satisfying $\partial_p\partial_{p+1}=0$. It is sometimes written as the following sequence of length $d$:
\begin{equation}\label{eq:chaincomplex}
  C_\bullet=[\cdots\to C_{p+1}\overset{\partial_{p+1}}\to C_p\overset{\partial_p}\to C_{p-1}\to\cdots].
\end{equation}
There are two distinguished subgroups: the groups of cycles $Z_p\df\ker{\partial_p}$ and the group of boundaries $B_p\df\im{\partial_{p+1}}$. Then the $p$-dimensional homology group of the chain complex $C_\bullet$ is the quotient $H_p(C_\bullet)\df Z_p/B_p$.

Similarly, one can define the \emph{cochain complex}
\begin{equation}\label{eq:coChaincomplex}
  C^\bullet=[\cdots\leftarrow C_{p+1}\overset{\d^{p+1}}\leftarrow C_p\overset{\d^p}\leftarrow C_{p-1}\leftarrow\cdots]
\end{equation}
and proceed as before by defining everything~\emph{co}:~$p$-cochains, $p$-coboundary maps, $p$-cocycles, etc, and consequently the $p$-th \emph{cohomology groups} $H^p(C^\bullet)\df\ker{\d^{p+1}}/\im{\d^p}$. In short, the coboundary map $\d^{p+1}$ raises the dimension of $C_p$ unlike $\partial^p$ which decreases it. Crucially, for a large class of spaces with certain properties there exists a powerful result called  \emph{Poincar\'e's duality}~\cite{munkres2018elements,bredon2013topology} relating the homology and cohomology groups:
\begin{equation}\label{eq:poincareDuality}
  H_p(C_\bullet)=H^{d-p}(C^\bullet).
\end{equation}

Focusing specifically on the regular CW complexes, the $p$-cycles of the CW complex $X$ are those $p$-chains $c_p$,~\cref{eq:dchains}, satisfying $\partial c_p=0$, where $\partial$ is a linear boundary map,~\cref{eq:boundaryMap}. $p$-cycles form an Abelian group (group of cycles) and inherit the module structure as well. The boundary $p$-cycles are the chains given by $\partial c_{p+1}$ and they again form a group (group of boundaries) and a ring-valued module. Thanks to~\cref{eq:parpar} I can therefore write $B_p(X,\bbZ_2)\subseteq Z_p(X,\bbZ_2)$ -- just like it is required from the axioms of the general theory of homology. Hence
$$
H_p(X,\bbZ_2)=Z_p(X,\bbZ_2)/B_p(X,\bbZ_2)
$$
is the $p$-th homology group of a CW~complex $X$. Further, there exists a topological invariant $\b_p(X)\df\rank{H_p(X,\bbZ_2)}$ called the $p$-th Betti number ($\rank{H_p}\equiv\dim{H_p}$ if the ring is $\bbZ_2$). Recalling the definitions of $Z_p$ and $B_p$ I write
\begin{equation}\label{eq:rankHp}
     \b_p(X)=\rank{\ker{\partial_p}}-\rank{\im{\partial_{p+1}}}.
\end{equation}

\section{Homological CSS codes from regular CW complexes}\label{app:homoCSS}

An $n$-bit linear code $C$ is defined by a binary matrix $G \in \bbZ_2^{k \times n}$ with $k \leq n$ linearly independent rows, referred to as the generator matrix. An element $x \in \bbZ_2^k$ is encoded by multiplying by $G$, i.e., $xG$. It is said that $C \subset C'$ if every codeword of $C$ is a codeword of $C'$. One can also define a parity check matrix (PCM) $H \in \bbZ_2^{(n-k)\times n}$ such that $H G^T = 0$, so that any encoded element is in the kernel of $H$. The parity check matrix also defines a code $C^\bot$, referred to as the dual of $C$, with generator matrix $H$ and parity check matrix $G$. One can define a CSS code from two $n$-bit linear codes $C_X$ and $C_Z$ with parity check matrices $H_X \in \bbZ_2^{r_X \times n}$ and $H_Z \in \bbZ_2^{r_Z \times n}$ satisfying
\begin{align}\label{eq:CSS}
    H_X H_Z^T = 0,
\end{align}
where $r_{X(Z)}=\rank{H_{X(Z)}}$. Let $h^i_X$ and $h^j_Z$ be the spanning bases of $C_X^\perp$ and $C_Z^\perp$. Then, $x_i=X^{h^i_X}$ are the $X$  type stabilizer  generators, $z_j=Z^{h^j_Z}$ are the $Z$~type stabilizer generators and~\cref{eq:CSS} expresses their mutual commutativity. Note that it is not uncommon for the bases of  $C_X^\perp$ or $C_Z^\perp$ to be overcomplete and so the corresponding $X(Z)$ stabilizers are not independent. So when I write, for example, $H_X$ as
\begin{equation}\label{eq:HXrowSpace}
  H_X=
  \begin{bmatrix}
   \cdots\ h^1_X\ \cdots \\[.8ex]
   \cdots\ h^2_X\ \cdots \\
    \vdots \\
   \cdots\ h^d_X\ \cdots\\
  \end{bmatrix}
\end{equation}
it will, in general, happen that $d>\rank{H_X}$.

How to obtain a pair satisfying~\cref{eq:CSS}? There exists many different strategies that have been developed since the inception of the CSS codes and it is an active area of research especially on the LDPC front. One of the well-explored ways is by employing the topological properties of various manifolds or just topological spaces like the CW complex homology over $\bbZ_2$ introduced in~\cref{appsub:CWhomology}. In fact, a triangulated or otherwise nicely cellulated (such as regularly in the sense or the regular CW complexes)  space yields a CSS code by virtue of its chain complex of length at least two and the corresponding boundary operators. Thus for any regular CW complex~$X$ introduced in~\cref{app:CWcomplexes}, one can construct a CSS code by setting $H_X = \partial_d$ and $H_Z^T = \partial_{d+1}$ and by using the incidence relations~\cref{eq:par} it is possible to find~\cref{eq:parpar}, so that~\cref{eq:CSS} holds. For any $v \in \im H_Z^T$, $Z^{v}$ is a check operator of the code, while for any $\nu \in  \ker H_X$, $Z^{\nu}$ commutes with all the $X$ checks of the code. Therefore the logical $Z$ operators correspond to the elements of $\ker H_X / \im H_Z^T$ and so the number of logical qubits is $\dim{[\ker H_X / \im H_Z^T]}$, which is the previously introduced Betti number in~\cref{eq:rankHp}.

The quotient space structure lends itself to an efficient way of finding the logical Pauli operations by lifting the basis vector to the original space. More precisely, let $V$ be a finite vector space and $U\subset V$. A homomorphism $\pi:V\mapsto V/U$ maps each $v\in V$ to a an element of a coset $V/U$, where $\ker{\pi}=U$. There exists a \emph{fixed lift homomorphism} $\varpi:V/U\mapsto V$ satisfying
\begin{equation}\label{eq:cosetIdentity}
  \pi(\varpi)=\id_{V/U}.
\end{equation}
So by setting $V=\ker H_X$ and $U=\im H_Z^T$ and finding the basis vectors I can lift them in a predictable way to the ambient space~$V$ (that's why it is fixed), where they become the logical Pauli operators $\{\ol X_i,\ol Z_i\}$ once their proper commutation relations are ensured. It is possible to ignore the lifting procedure altogether and follow the algorithm introduced~in~\cite{wilde2009logical}, which by using the knowledge of the $X$ and $Z$ generator matrices $G_X,G_Z$ of a CSS code iteratively generates logical operators while removing the stabilizer generators.

From the PCMs, one can construct a convenient representation of the CSS codes (length two chain complexes in the case of homological CSS codes) in the form of a \emph{Tanner graph}. The Tanner graph is a bipartite graph, where the bipartition divides the $X,Z$ checks and the data qubits. It is drawn in the form of a tripartite graph corresponding to two incidence matrices (a layer of $X$ checks  that act nontrivially on data qubits and the data connected in a similar manner to the $Z$ checks).

\section{Regular CW complexes and the Bruhat order of Coxeter groups}\label{app:sphereClassification}

\begin{figure}[t]
  \resizebox{14cm}{!}{\includegraphics{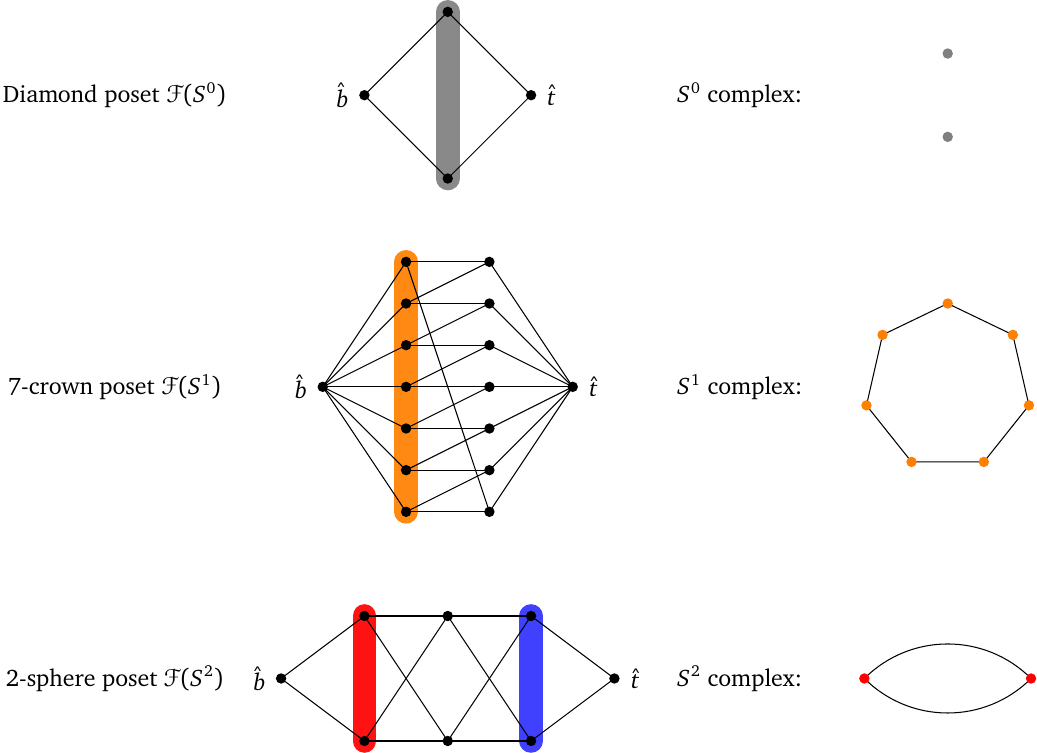}}
  \caption{The poset structure and the corresponding regular CW complexes of the basic building blocks of the CSS codes constructed here. On the left, the diamond poset $\euF(S^0)$, $k$-crown $\euF(S^1)$ for $k=7$ and $\euF(S^2)$ of probably the simplest regular complex of $S^2$. On the right, the CW complexes with a color coding to help identify which 0-,1- or 2-cells are to be found. The two 2-cells of $S^2$ correspond to the inner and outer faces (not depicted in blue).}
  \label{fig:S012}
\end{figure}

Everything is now ready to show how Coxeter groups and their Bruhat order are related to the regular CW complexes (and therefore CSS codes as elaborated on in the main text). The key result is due to Bj\"orner~\cite{bjorner1984posets,bjorner2006combinatorics}:
\begin{thm}\label{thm:CWcellulationBruhat}
  Let $(W,S)$ be a Coxeter system and let $(u,v)$ be an open interval of the Bruhat order for $u,v\in W$ such that $\ell(u,v)\geq2$. Then $(u,v)$ is isomorphic to the face poset of a regular CW complex of $S^{\ell(u,v)-2}$.
\end{thm}
This is an interesting result for several reasons. Given its relevance for us I will just rephrase it in words before looking closer at the structure of open intervals.   For a finite or infinite Coxeter group one can choose two elements called  bottom and top, $w_b$ and $w_t$, as is common in the poset terminology. For small finite Coxeter groups a natural choice is the unique shortest and longest elements $w_b=\id$ and $w_t$.  Because a poset vertex representing a reduced word $w = s_{i_1}\ldots s_{i_p}\in W$ is naturally associated with a $(p-1)$-cell in this way  I get a regular CW complex of a $(d-2)$-dimensional sphere if $\ell(w_t)=d$. But that is not necessary or even possible for infinite Coxeter groups. In fact, in both cases I may simply choose any two elements $u,v\in W$ to be $w_b$ and $w_t$ as long as they are comparable, i.e. $u<v$, and differ in length by at least two. The corresponding Bruhat poset is again a face poset of cells of a regular CW complex cellulating $S^{d-2}$ -- but  this time for an arbitrarily large~$d$ if I consider a big enough Coxeter group.

This all suggests an immense variability of  regular CW cellulations of $S^d$ spheres one can obtain even considering just irreducible Coxeter groups (whenever such classification is available such as finite groups, tessellations and certain classes of hyperbolic groups). But surprisingly, there is a hidden structure for all closed intervals $[w_b,w_t]$ of length~$d+2, d\geq0$. Moreover, for certain (small) lengths and certain classes of Coxeter groups their building blocks have been classified~\cite{bjorner2006combinatorics}.

\begin{figure}[t]
  \resizebox{14cm}{!}{\includegraphics{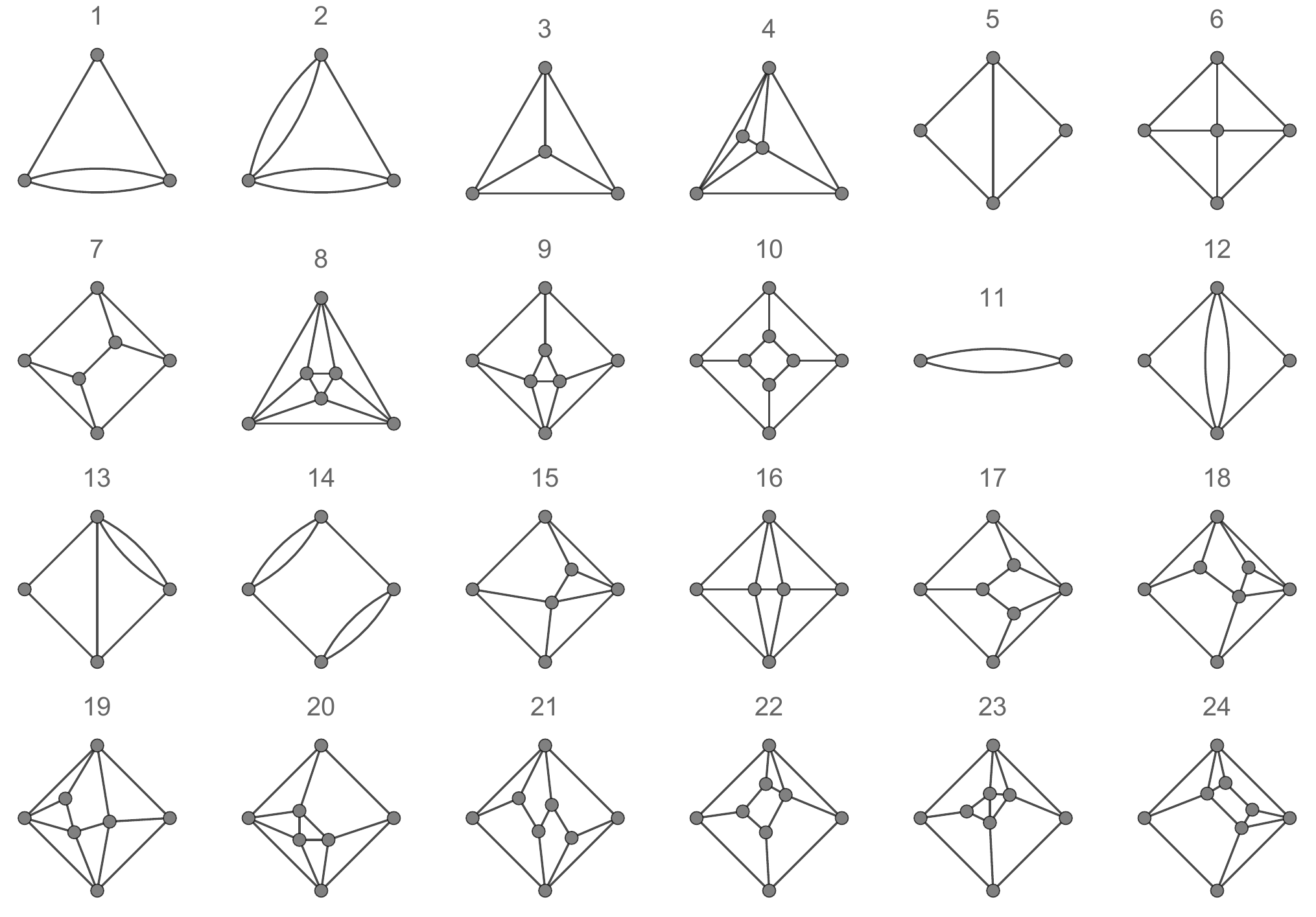}}
  \caption{Hultman's classification of the length four face posets $\euF(S^2)$ of for all finite Weyl groups~\cite{hultman2003combinatorial} in terms of the 24 $S^2$ regular CW complexes. Vertices are 0-cells, edges are 1-cells and faces (including an `ambient' cell for each CW complex because the figures are stereographic projections) are 2-cells.}
  \label{fig:hultman}
\end{figure}

Before summarizing of what is known it is helpful to depict the face poset structure of the three lowest dimensional spheres as the main tools used in this paper, see~\cref{fig:S012}. For $S^0$, being the boundary of a line, the face poset of the $S^0$ CW complex is simply two points (two 0-cells). If a common bottom and top element $\hat{b}$ and $\hat{t}$ are added I will call the resulting 3-layer poset the \emph{diamond} poset (graph). For $S^1$, the number of CW complexes is infinite but they are easy to characterize. They are all polygons with $k$ sides for $k\geq2$ (also called $k$-gons) so except for $k=2$ they are all polyhedral complexes (a special type of CW). An $S^1$ face poset $\euF(S^1)$ consists of $k$ 0-cells and $k$ 1-cells. The face-incident relations therefore result in a  $2k$-cycle -- a potentially useful observation for the QEC code properties. As a graded poset with a bottom and top element added its Hasse diagram is depicted in the middle row of~\cref{fig:S012} resulting in a 4-layer poset. $\euF(S^1)$ was nicknamed a $k$-crown in~\cite{bjorner2006combinatorics} and I use this terminology extensively in the main text. Finally, for $S^2$, a generic CW complex consisting of 0-,1- and 2-cells  takes form of a planar graph, which is possible because every cellulated $S^2$ can be stereographically projected onto a two-dimensional plane (bottom row of~\cref{fig:S012}). I depict a simple example already constructed earlier in~\cref{fig:S2CW} (left).  Again, a common bottom and top  element is added promoting it to a 5-layer poset. While the $S_1$ complex is a polyhedral complex, the $S^2$ complex is a bona fide CW complex as can be seen from its 1-skeleton which is a digon.

The aforementioned small interval classification is summarized as follows:
\begin{prop}\label{prop:sphereClassification}
Let $(W,S)$ be a Coxeter system and $u,v\in W$ such that $u<v$. Then
\begin{enumerate}
  \item All intervals $[u,v]$ of length two  are diamond posets,
  \item All $\ell([u,v])=3$ intervals are isomorphic to $k$-crown posets,
  \item All $\ell([u,v])=4$ intervals are isomorphic to one of the 24 spheres classified in~\cite{hultman2003combinatorial,hultman2003bruhat} if $W$ is a finite Weyl group.
\end{enumerate}
\end{prop}
The 24  regular CW $S^2$ spheres are depicted in~\cref{fig:hultman}. As planar graphs, there are pairs of dual graphs (such as 12 and 14) and specimens of self-dual graphs such as number 11 (which is also my $S^2$ example from~\cref{fig:S2CW} (in 3D)).

The Weyl group classification is not an exhaustive list of all possible $S^2$ CW complexes. The reason is another result from~\cite[p.~52, Example~2.7.9]{bjorner2006combinatorics} showing that all possible $k$-crowns can appear for infinite-dimensional Coxeter groups such as the hyperbolic triangle group family,~\cref{eq:triangleGroup}. Since the $k$-crowns are the $k$-gons appearing in the 1-skeleta of the $S^2$ complexes (cf.~\cref{fig:hultman}) it follows that the list is not complete (by inspecting that there are no vertices or faces (dual vertices)) of weight five and more.

\bibliographystyle{unsrt}

%\bibliography{bruhat.bib}

\end{document}